\newcommand{\longversion}[1]{#1}
\newcommand{\shortversion}[1]{}

\shortversion{%
\documentclass[]{svproc}%
}
\longversion{%
  \documentclass[10pt,letter]{article}%
  \usepackage[hscale=0.7,scale=0.75]{geometry}
  \pdfoutput=1
}
\usepackage[USenglish]{babel}%
\newcommand{\footnoteitext}[1]{\stepcounter{footnote}
  \footnotetext[\thefootnote]{#1}}

\usepackage{multibib}
\newcites{sec}{Appendix References}

\def\etal{et~al.{}}

\usepackage{amssymb,amsfonts,amsmath}%
\usepackage{cases}
\DeclareFontFamily{U}{matha}{\hyphenchar\font45}
\DeclareFontShape{U}{matha}{m}{n}{
      <5> <6> <7> <8> <9> <10> gen * matha
      <10.95> matha10 <12> <14.4> <17.28> <20.74> <24.88> matha12
      }{}
\DeclareSymbolFont{matha}{U}{matha}{m}{n}
\DeclareMathSymbol{\squplus}{2}{matha}{"5D}

\makeatletter
\newcommand{\raisemath}[1]{\mathpalette{\raisem@th{#1}}}
\newcommand{\raisem@th}[3]{\raisebox{#1}{$#2#3$}}
\makeatother

\usepackage{stmaryrd}%

\makeatletter
\newcommand{\pushright}[1]{\ifmeasuring@#1\else\omit\hfill\ensuremath{\displaystyle#1}\fi\ignorespaces}
\newcommand{\pushleft}[1]{\ifmeasuring@#1\else\omit$\displaystyle#1$\hfill\fi\ignorespaces}
\providecommand{\leftsquigarrow}{%
	\mathrel{\mathpalette\reflect@squig\relax}%
}
\newcommand{\reflect@squig}[2]{%
	\reflectbox{$\m@th#1\rightsquigarrow$}%
}

\makeatother
\DeclareMathOperator{\disj}{DISJ}%
\DeclareMathOperator{\choice}{CH}%
\DeclareMathOperator{\weight}{WGT}%
\DeclareMathOperator{\opt}{OPT}%

\DeclareMathOperator{\wght}{wght}%
\DeclareMathOperator{\cst}{cst}%
\DeclareMathOperator{\cnt}{cnt}%
\DeclareMathOperator{\bnd}{bnd}%
\DeclareMathOperator{\cntc}{\#\cdot}%
\DeclareMathOperator{\optimize}{\leftsquigarrow}

\DeclareMathOperator{\type}{type}
\newcommand{\intr}{\textit{int}}
\newcommand{\leaf}{\textit{leaf}}
\newcommand{\rem}{\textit{rem}}
\newcommand{\join}{\textit{join}}

\DeclareMathOperator{\Mod}{Mod}

\DeclareMathOperator{\SSR}{SatRules}

\DeclareMathOperator{\UpdateStates}{UpdtWgt}

\DeclareMathOperator{\UpdateRedStates}{UpdtWgt\&Ch}
\DeclareMathOperator{\post}{post-order}

\DeclareMathOperator{\kmin}{kmin}

\newcommand{\Tab}[1]{\ensuremath{\text{Child-Tabs}}}
\newcommand{\Tabs}[1]{\ensuremath{\text{Tables[#1]}}}
\renewcommand{\Big}{} 

\newcommand{\CCC}{\ensuremath{\mathcal{C}}}%
\newcommand{\SSS}{\ensuremath{\mathcal{S}}}%
\newcommand{\TTT}{\ensuremath{\mathcal{T}}}%

\usepackage{latexsym}
\usepackage{rotating}

\usepackage{paralist}
\usepackage{enumerate}
\setdefaultleftmargin{10pt}{}{}{}{}{}

\usepackage{url}
\usepackage{xspace}

\newcommand{\ASP}{\textsc{Asp}\xspace}

\newcommand{\QBF}{\textsc{Qbf}\xspace}
\newcommand{\pname}[1]{\textsc{#1}\xspace}

\newcommand*\mcup{\mathbin{\mathpalette\mcupinn\relax}}
\newcommand*\mcupinn[2]{\vcenter{\hbox{$\mathsurround=0pt
  \ifx\displaystyle#1\textstyle\else#1\fi\bigcup$}}}
\usepackage[]{xcolor}

\newcommand{\emptyfunc}{\emptyset}
\newcommand{\NAT}{\ensuremath{\mathbb{N}}}

\newcommand{\inputPredColor}{orange!55!red}
\newcommand{\outputPredColor}{blue!45!black}
\newcommand{\statePredColor}{green!62!black}

\usepackage{multirow}

\usepackage{graphicx}
\graphicspath{{fig/}}
\DeclareGraphicsExtensions{.pdf,.png,.jpg}
\usepackage{booktabs}

\usepackage{microtype}
\usepackage{algpseudocode}
\algrenewcommand\algorithmicensure{\textbf{Output:}}
\algrenewcommand{\algorithmiccomment}[1]{\emph{// #1}}

\usepackage[ruled,vlined,linesnumbered]{algorithm2e}
\SetKwInput{KwData}{In}
\SetKwInput{KwResult}{Out}
\SetAlFnt{\small}
\SetAlCapFnt{\small}
\SetAlCapNameFnt{\small}
\SetAlCapHSkip{0pt}
\SetEndCharOfAlgoLine{}
\IncMargin{-\parindent}

\usepackage{paralist}

\usepackage{float}
\usepackage[hang,labelsep=quad,aboveskip=2pt]{caption}

\usepackage{subcaption}
\usepackage{csquotes}
\usepackage{paralist}
\usepackage{multicol}
\usepackage{multirow}
\usepackage{listings}

\newcommand{\tuplecolor}[1]{\textcolor{#1}}
\lstdefinelanguage{dflat}{
	numberstyle=\tiny,
	otherkeywords={:-},
	morekeywords={not},
	keywordstyle=\bfseries,
	emph={numChildNodes,initial,final,currentNode,childNode,bag,current,introduced,removed,atLevel,atNode,root,rootOf,leaf,leafOf,sub,childItem,childAuxItem,childCost,childOr,childAnd,childAccept,childReject,childRow},
	moreemph=[2]{item,auxItem,extend,cost,currentCost,length,or,and,accept,reject},
	alsoletter={\#}, 
	morecomment=[l]{\%},
	emphstyle=\color{\inputPredColor},
	emphstyle=[2]\color{\outputPredColor},
	literate={:-}{{$\leftarrow$}}2 {!=}{{$\neq$}}1, 
	breakindent=3em,
	escapechar=@, 
	captionpos=b,
        frame=bt, 
        numbers=left
}
 \usepackage{relsize}
\usepackage{tikz}
\usetikzlibrary{shapes}
\usetikzlibrary{calc}
\longversion{%
  \usepackage{amsthm}
  \newtheorem{remark}{Remark}%
  \newtheorem{theorem}{Theorem}%
  \newtheorem{lemma}{Lemma}%
  \newtheorem{proposition}{Proposition}%
  \newtheorem{example}{Example}%
  \newtheorem{definition}{Definition}%
  \newtheorem{corollary}{Corollary}%
}

\newtheorem{observation}{Observation}

\shortversion{%
  \renewenvironment{example}{\begin{EXa}}{\hfill\ensuremath{\blacksquare}\end{EXa}}
  \spnewtheorem{EXa}{Example}{\bfseries}{\normalfont}
}

\newcommand{\MAI}[1]{\ensuremath{#1^+_a}}%
\newcommand{\MAIRR}[2]{\ensuremath{#1^+_{#2}}}%
\newcommand{\MAR}[1]{\ensuremath{#1^-_a}}%
\newcommand{\MARRR}[2]{\ensuremath{#1^-_{#2}}}%

\newcommand{\MARR}[2]{\ensuremath{#1^-_{#2}}}%
\newcommand{\MAIR}[2]{\ensuremath{#1^+_{#2}}}%
\newcommand{\tabval}{\ensuremath{u}}
\newcommand{\tab}[1]{\ensuremath{\tau_{#1}}}
\newcommand{\at}{\text{\normalfont at}}
\newcommand{\oo}{\text{\normalfont cst}}
\newcommand{\att}[1]{\ensuremath{\at_{\hspace{-0.05em}\leq\hspace{-0.05em}#1}}}
\newcommand{\atto}{\ensuremath{\att{t}}}
\newcommand{\progt}[1]{\ensuremath{\prog_{\hspace{-0.05em}\leq\hspace{-0.05em}#1}}}
\newcommand{\progtneq}[1]{\ensuremath{\prog_{\hspace{-0.05em}<\hspace{-0.05em}#1}}}

\newcommand{\dpa}{\ensuremath{\mathcal{DP}}}
\newcommand{\por}{\vee}

\newcommand{\eqdef}{\ensuremath{\,\mathrel{\mathop:}=}}
\newcommand{\hsep}{\leftarrow\,}
\newcommand{\Card}[1]{|#1|}

\newcommand{\AspCons}{\pname{Cons}} %
\newcommand{\AspComp}{\pname{AS}} %
\newcommand{\AspCount}{\pname{\#Asp}} %
\newcommand{\AspCountO}{\pname{\AspCount{}O}} %

\newcommand{\PRIMSAT}{\ensuremath{{\algo{MOD}}}\xspace}
\newcommand{\PRIM}{\ensuremath{{\algo{PRIM}}}\xspace}

\newcommand{\INCSAT}{\ensuremath{{\algo{IMOD}}}\xspace}
\newcommand{\INC}{\ensuremath{{\algo{INC}}}\xspace}
\newcommand{\problemFont}[1]{\textsc{#1}}

\newlength\problemlength
\newcommand\dproblem[3]{%
\begin{center}
\fbox{%
\begin{minipage}{.93	\linewidth}%
\begin{list}{}{\labelwidth\problemlength \labelsep.7em \rightmargin1.5em
\leftmargin\problemlength \advance\leftmargin by3em
\parsep0ex \itemsep.2ex plus.1ex}
\item[{\sl Problem:\hfill}] {\problemFont{#1}}
\item[{\sl Input:  \hfill}] #2
\item[{\sl Task: \hfill}] #3
\end{list}
\end{minipage}
}
\end{center}
}

\usepackage{url}

\usepackage{enumitem}

\newenvironment{indented}{\begin{changemargin}{1cm}{0cm}}{\end{changemargin}}

\let\phi\varphi
\let\epsilon\varepsilon

\renewcommand{\models}{\vDash}

\newcommand{\calT}{\mathcal{T}}

\newcommand{\card}[1]{\left|#1\right|}
\newcommand{\CCard}[1]{\|#1\|}

\newcommand{\algo}[1]{\ensuremath{\mathsf{#1}}}

\newcommand{\NP}{\ensuremath{\textsc{NP}}\xspace}
\newcommand{\co}{\ensuremath{\textsc{co}}}

\newcommand{\bigO}[1]{\ensuremath{{\mathcal O}(#1)}}

\newcommand{\tw}[1]{\mathit{tw}(#1)}

\newcommand{\SB}{\{}%
\newcommand{\SM}{\mid}%
\newcommand{\SE}{\}}%
\def\hy{\hbox{-}\nobreak\hskip0pt} 

\newcommand{\solver}[1]{\mbox{\text{#1}}\xspace}

\newcommand{\depqbfz}{\solver{DepQBF0}}
\newcommand{\sharpsat}{\solver{SharpSAT}}
\newcommand{\dynasp}[1]{\ensuremath{\solver{DynASP2}(#1)}}

\newcommand{\prog}{\ensuremath{\Pi}}

\usetikzlibrary{positioning}
\usetikzlibrary{shapes,arrows}
\tikzstyle{arg}=[draw, thick, circle]

\colorlet{afnodecolor}{green!20!blue!10}
\colorlet{tdnodecolor}{green!20!blue!10}
\colorlet{subfwnodecolor}{black!2}
\colorlet{subfwafinactivenodecolor}{white}
\colorlet{vertexTopColor}{white}
\colorlet{vertexBottomColor}{black!10}

\tikzstyle{afnode} = [draw,thick,shape=circle,minimum size=8mm,font=\normalsize,fill=afnodecolor]
\tikzstyle{afedge} = [->,draw,thick]
\tikzstyle{tdnode} = [draw,rounded corners,top color=vertexTopColor,bottom color=vertexBottomColor,minimum size=1.5em]
\tikzstyle{stdnode} = [tdnode, font=\scriptsize]
\tikzstyle{stdnodecompact} = [stdnode, inner sep = 1.5pt, outer sep = 0.1pt]
\tikzstyle{stdnodetable} = [stdnode, inner sep = 1.5pt, outer sep = 0]
\tikzstyle{stdnodenum} = [minimum size=1.5em, font=\scriptsize]
\tikzstyle{tdedge} = [-,draw,thick]
\tikzstyle{tdlabel} = [draw=none, rectangle, fill=none, inner sep=0pt, font=\scriptsize]
\tikzstyle{subfwnode} = [draw,thick,shape=rectangle,thin,rounded corners,minimum size=9mm,fill=subfwnodecolor,label distance=-2.5mm]
\tikzstyle{subfwafactivenode} = [draw,thick,shape=circle,minimum size=6mm,inner sep = 0pt,font=\scriptsize,fill=afnodecolor]
\tikzstyle{subfwafinactivenode} = [draw,thick,shape=circle,minimum size=6mm,inner sep = 0pt,font=\scriptsize,fill=white,dotted]
\tikzstyle{subfwafinactiveedge} = [->,draw,thick,dotted]

\tikzstyle{itemTree}=[level distance=2em,sibling distance=4ex,child anchor=west,grow'=right,right,align=left,every node/.style={draw,dashed,draw opacity=0.2,font=\footnotesize}]
\tikzstyle{itemTreeRoot}=[solid,inner sep=2]
\tikzstyle{orNode}=[label=left:$\lor$]
\tikzstyle{andNode}=[label=left:$\land$]
\tikzstyle{acceptNode}=[label=right:$\top$]
\tikzstyle{rejectNode}=[label=right:$\bot$]

\usepackage[hidelinks]{hyperref}

\hypersetup{pdfinfo={
  Title={Answer Set Solving with Bounded Treewidth Revisited}
  Author={Johannes~K.~Fichte, Markus Hecher, Michael Morak, Stefan Woltran}
}}

\begin{document}
\shortversion{%
\mainmatter

\title{
	Answer Set Solving with \\ Bounded Treewidth Revisited}
\titlerunning{
	Answer Set Solving with Bounded Treewidth Revisited}  
%
\author{Johannes~K.~Fichte \and Markus Hecher \and\\
Michael Morak \and Stefan Woltran}
\authorrunning{Fichte et al.} 
%
\tocauthor{Johannes~K.~Fichte, Michael Morak, Markus Hecher and Stefan Woltran}
\institute{TU Wien, Austria\\
\email{lastname@dbai.tuwien.ac.at}}

} 
\longversion{%
  \title{Answer Set Solving with Bounded Treewidth
    Revisited\thanks{This is the author’s self-archived copy including
      detailed proofs. A preliminary version of the paper was
      presented on the workshop TAASP'16. Research was supported by
      the Austrian Science Fund (FWF), Grant Y698.}}
  \author{Johannes~K.~Fichte\footnote{Also
      affiliated with the Institute of Computer Science and
      Computational Science at
      University of Potsdam, Germany.}, Michael Morak, Markus Hecher and Stefan Woltran\\[3pt]
    TU Wien, Austria\\
    lastname@dbai.tuwien.ac.at }
  \date{}%
}

\maketitle              
\begin{abstract}
  Parameterized algorithms are a way to solve hard problems more
  efficiently, given that a specific parameter of the input is
  small. In this paper, we apply this idea to the field of answer set
  programming (ASP).  To this end, we propose two kinds of graph
  representations of programs to exploit their treewidth as a
  parameter. Treewidth roughly measures to which extent the internal
  structure of a program resembles a tree. Our main contribution is
  the design of parameterized dynamic programming algorithms, which
  run in linear time if the treewidth and weights of the given program
  are bounded. Compared to previous work, our algorithms handle the
  full syntax of ASP.  Finally, we report on an empirical evaluation
  that shows good runtime behaviour for benchmark instances of low
  treewidth, especially for counting answer sets.

\shortversion{%
}
\end{abstract}

\section{Introduction}
Parameterized algorithms~\cite{Niedermeier06,CyganEtAl15} have
attracted considerable interest in recent years and allow to tackle
hard problems by directly exploiting a small parameter of the input
problem. One particular goal in this field is to find guarantees that
the runtime is exponential exclusively in the parameter, and polynomial
in the input size (so-called fixed-parameter tractable algorithms). A parameter
that has been researched extensively is
treewidth~\cite{RobertsonSeymour86,BodlaenderKoster08}.  Generally
speaking, treewidth measures the closeness of a graph to a tree, based
on the observation that problems on trees are often easier than on
arbitrary graphs. A parameterized algorithm exploiting small treewidth
takes a tree decomposition, which is an arrangement of a graph into a
tree, and evaluates the problem in parts, via dynamic programming (DP)
on the tree decomposition. 

ASP~\cite{BrewkaEiterTruszczynski11,Lifschitz08} is a logic-based
declarative modelling language and problem solving framework where
solutions, so called answer sets, of a given logic program directly
represent the solutions of the modelled
problem. Jakl~\etal~\cite{JaklPichlerWoltran09} give a DP algorithm
for disjunctive rules only, whose runtime is linear in the input size
of the program and double exponential in the treewidth of a particular
graph representation of the program structure.
However, modern ASP systems allow for an extended syntax that
includes, among others, weight rules and choice
rules. Pichler~\etal~\cite{PichlerEtAl14} investigated the complexity of
programs with weight rules. They also presented DP algorithms
for programs with cardinality rules (i.e., restricted version of
weight rules), but without disjunction.

In this paper, we propose DP algorithms for finding answer sets that
are able to directly treat all kinds of ASP rules. While such rules can be
transformed into disjunctive rules, we avoid the resulting polynomial
overhead with our algorithms.
In particular, we present two approaches based on two different types
of graphs representing the program structure. Firstly, we consider 
the primal graph, which allows for an intuitive algorithm that also treats
the extended ASP rules.
While for a given disjunctive program the treewidth of the primal graph may
be larger than treewidth of the graph representation used by
Jakl~\etal~\cite{JaklPichlerWoltran09}, our algorithm uses
simpler data structures and lays the foundations to understand how we can
handle also extended rules.
Our second graph representation is the incidence graph, a
generalization of the representation used by Jakl~\etal.  Algorithms
for this graph representation are more sophisticated, since weight and
choice rules can no longer be completely evaluated in the same
computation step.  Our algorithms yield upper bounds that are linear
in the program size, double-exponential in the treewidth, and
single-exponential in the maximum weights. We extend two algorithms to
count optimal answer sets. For this particular task, experiments show
that we are able to outperform existing systems from multiple domains,
given input instances of low treewidth, both randomly generated and
obtained from real-world graphs of traffic networks.  Our system is
publicly available~on~github\footnote{See
  \url{https://github.com/daajoe/dynasp}.}.






\section{Formal Background} %
\label{sec:preliminaries}%
\shortversion{\noindent \textit{Answer Set programming (ASP).} }%
\longversion{\subsection{Answer Set programming (ASP)}}%
\emph{ASP} is a declarative modeling and problem solving framework;
for a full introduction, see,~e.g.,
\cite{BrewkaEiterTruszczynski11,Lifschitz08}.
State-of-the-art ASP grounders support the full ASP-Core-2
language~\cite{aspcore2} and output smodels input
format~\cite{lparse}, which we will use for our algorithms.
Let $\ell$, $m$, $n$ be non-negative integers such that
$\ell \leq m \leq n$, $a_1$, $\ldots$, $a_n$ distinct propositional
atoms, $w$, $w_1$, $\ldots$, $w_n$ non-negative integers, and
$l \in \{a_1, \neg a_1\}$.
A \emph{choice rule} is an expression of the form, 
$\{a_1; \ldots; a_\ell \} \hsep a_{\ell+1}, \ldots, a_m, \neg a_{m+1},
\ldots, \neg a_n$,
%
%
%
a \emph{disjunctive rule} is of the form
%
$a_1\por \cdots \por a_\ell \hsep a_{\ell+1}, \ldots, a_{m}, \neg
a_{m+1}, \ldots, \neg a_n$ and
%
%
%
%
%
a \emph{weight rule} is of the form
%
$a_{\ell} \hsep w \leqslant \{ a_{\ell + 1} = w_{\ell + 1}, \ldots, a_m = w_m,
\, \neg a_{m+1} = w_{m+1}, \ldots, \neg a_n = w_n \}$.
%
%
%
Finally, an
\emph{optimization rule} is an expression of the form
%
%
$\optimize l[w]$. 
%
%
%
%
A \emph{rule} is either a disjunctive, a choice, a weight, or an
optimization rule.
\longversion{
}
%
For a choice, disjunctive, or weight rule~$r$, let
$H_r \eqdef \{a_1, \ldots, a_\ell\}$, 
%
$B^+_r \eqdef \{a_{\ell+1}, \ldots, a_{m}\}$,
%
and $B^-_r \eqdef \{a_{m+1}, \ldots, a_n\}$.
%
%
%
For a weight rule~$r$, let $\wght(r,a)$ map atom $a$ to its
corresponding weight~$w_i$ in rule~$r$ if $a=a_i$ for
$\ell+1 \leq i \leq n$ and to $0$ otherwise, let
$\wght(r,A) \eqdef \sum_{a \in A} \wght(r,a)$ for a set $A$ of atoms,
and let $\bnd(r)\eqdef w$ be its \emph{bound}.
For an optimization rule~$r$, let $\cst(r) \eqdef w$ and if $l=a_1$,
let $B^+_r\eqdef \{a_1\}$ and $B^-_r\eqdef \emptyset$; or if
$l=\neg a_1$, let $B^-_r\eqdef \{a_1\}$ and $B^+_r\eqdef \emptyset$.
For a rule $r$, let $\at(r) \eqdef H_r \cup B^+_r \cup B^-_r$ denote
its \emph{atoms} and
$B_r \eqdef B^+_r \cup \SB \neg b \SM b \in B^-_r \SE$ 
its \emph{body}.
A \emph{program}~$\prog$ is a set of rules. Let $\at(\prog) \eqdef \SB
\at(r) \SM r \in \prog \SE$
and let $\choice(\prog), \disj(\prog),
\opt(\prog)$ and $\weight(\prog)$ denote the set of all choice, disjunctive,
optimization and weight rules in~$\prog$, respectively.
\longversion{%
} 
A set $M \subseteq \at(\prog)$ \emph{satisfies} a rule~$r$ if
(i)~$(H_r \cup B^-_r) \cap M \neq \emptyset$ or $B^+_r \not\subseteq
M$ for~$r\in \disj(\prog)$,
(ii)~$H_r \cap M \neq \emptyset$ or $\Sigma_{a_i \in M \cap B^+_r} \;
w_i + \Sigma_{a_i \in B^-_r \setminus M} \; w_i< \bnd(r)$ for~$r \in
\weight(\prog)$, or 
(iii) $r\in\choice(\prog)\cup \opt(\prog)$.
%
$M$ is a model of~$\prog$, denoted by $M \models \prog$, if $M$
satisfies every rule~$r \in \prog$. 
Further, let
$\Mod(\CCC,\prog) \eqdef \SB C \SM C \in \CCC, C \models \prog \SE$
for $\CCC\subseteq 2^{\at(\prog)}$.
\longversion{%
}
The \emph{reduct}~$r^M$  
(i)~of a choice rule~$r$ is the set
$\SB a \leftarrow B^+_r \SM a \in H_r \cap M, B^-_r \cap M =
\emptyset\SE$ of rules,
(ii)~of a disjunctive rule~$r$ is the singleton
$\SB H_r \leftarrow B^+_r \SM B^-_r \cap M = \emptyset\SE$, and
(iii)~of a weight rule~$r$ is the singleton
$\SB H_r \leftarrow w' \leqslant [ a = \wght(r, a) \SM a \in B^+_r]\SE$
where
%
%
$w' = \bnd(r) - \allowbreak \Sigma_{a \in B^-_r \setminus M}\wght(r,
a)$.
$\prog^M := \SB r' \SM r' \in r^M, r \in \prog \SE$ is called \emph{GL
  reduct} of $\prog$ with respect to~$M$.
%
%
%
%
%
%
%
%
%
%
A set~$M \subseteq \at(\prog)$ is an \emph{answer set} of
program~$\prog$ if (i) 
 $M \models \prog$
and (ii) there is no
$M' \subsetneq M$ such that $M' \models \prog^M$, that is,
$M$ is \emph{subset minimal with respect to $\Pi^M$}.
\longversion{%
}%
We call
$\oo(\prog,M,A)\eqdef\allowbreak \Sigma_{r \in \opt(\prog),\ A \cap
  [(B^+_r \cap M) \cup (B^-_r \setminus M)] \neq \emptyset} \cst(r)$
the \emph{cost} of model $M$ for $\prog$ with respect to the set
$A \subseteq \at(\prog)$. An answer set~$M$ of $\prog$ is
\emph{optimal} if its cost is minimal over all answer sets.
%
%
%
%
%
%
\begin{figure}[t]%
\vspace{-1em}
\centering
\hspace{-2em}\begin{tikzpicture}[node distance=7mm,every node/.style={fill,circle,inner sep=2pt}]
\node (a) [label={[text height=1.5ex,yshift=0.0cm,xshift=0.05cm]left:$d$}] {};
\node (b) [right of=a,label={[text height=1.5ex]right:$a$}] {};
\node (c) [below left of=b,label={[text height=1.5ex,yshift=0.09cm,xshift=0.05cm]left:$c$}] {};
\node (d) [below right of=b,label={[text height=1.5ex,yshift=0.09cm,xshift=-0.05cm]right:$b$}] {};
\draw (a) to (b);
\draw (b) to (c);
\draw (b) to (d);
\draw (c) to (d);
\end{tikzpicture}\hspace{-0.5em}%
\begin{tikzpicture}[node distance=0.5mm]
\tikzset{every path/.style=thick}

\node (leaf1) [tdnode,label={[yshift=-0.25em,xshift=0.5em]above left:$t_1$}] {$\{a,b,c\}$};
\node (leaf2) [tdnode,label={[xshift=-1.0em, yshift=-0.15em]above right:$t_2$}, right = 0.1cm of leaf1]  {$\{a,d\}$};
\coordinate (middle) at ($ (leaf1.north east)!.5!(leaf2.north west) $);
\node (join) [tdnode,ultra thick,label={[]left:$t_3$}, above  = 1mm of middle] {$\{a\}$};

\coordinate (top) at ($ (join.north east)+(3.5em,0) $);
\coordinate (bot) at ($ (top)+(0,-4em) $);

\draw [dashed] (top) to (bot);
\draw [<-] (join) to (leaf1);
\draw [<-] (join) to (leaf2);
\end{tikzpicture}\hspace{-0.0em}%
\begin{tikzpicture}[node distance=7mm,every node/.style={fill,circle,inner sep=2pt}]
\node (a) [label={[text height=1.5ex,yshift=0.0cm,xshift=0.12cm]left:$d$}] {};
\node (b) [right = 0.5cm of a,label={[text height=1.5ex,xshift=0.12cm]left:$a$}] {};
\node (c) [right = 0.5cm of b,label={[text height=1.5ex,xshift=0.12cm]left:$b$}] {};
\node (d) [right = 0.5cm of c,label={[text height=1.5ex,xshift=-0.05cm]right:$c$}] {};
\node (r3) [below = 0.5cm of a,label={[text height=1.5ex]right:${r_3}$}] {};
\node (r1) [below = 0.5cm of c,label={[text height=1.5ex,xshift=0.12cm]left:${r_1}$}] {};
\node (r2) [below = 0.5cm of d,label={[text height=1.5ex,xshift=-0.05cm]right:${r_2}$}] {};
\draw (a) to (r3);
\draw (b) to (r3);
\draw (b) to (r1);
\draw (c) to (r1);
\draw (d) to (r1);
\draw (b) to (r2);
\draw (c) to (r2);
\draw (d) to (r2);
\end{tikzpicture}\hspace{-0.5em}%
\begin{tikzpicture}[node distance=0.5mm]
\tikzset{every path/.style=thick}

\node (leaf0) [tdnode,label={[]left:$t_1$}] {$\{b,c, {r_1}, {r_2}\}$};
\node (leaf1) [tdnode,label={[xshift=0em]left:$t_2$}, above = 0.1cm of leaf0] {$\{a,{r_1}, {r_2}\}$};
\node (leaf2) [tdnode,label={[xshift=0em]above right:$t_3$}, right = 0.1cm of leaf0]  {$\{a, d, {r_3}\}$};
\coordinate (middle) at ($ (leaf1.north east)!.5!(leaf2.north west) $);
\node (join) [tdnode,ultra thick,label={[xshift=-0.25em]right:$t_4$}, right = 0.15cm of leaf1] {$\{a\}$}; 

\draw [->] (leaf0) to (leaf1);
\draw [<-] (join) to (leaf1);
\draw [<-] (join) to (leaf2);
\end{tikzpicture}%
%
\caption{Graph~$G_1$ with a TD of $G_1$ \hspace{-0.05cm}(left) and graph~$G_2$ with a TD
  of $G_2$ \hspace{-0.05cm}(right).}%
\label{fig:graph-td}%
\end{figure}
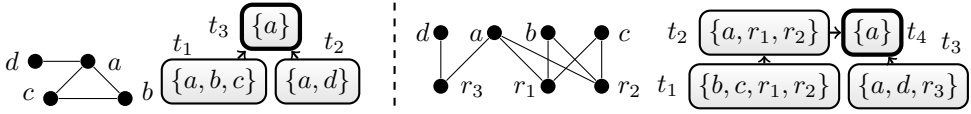
\longversion{%
}%
\begin{example}\label{ex:running1}
  Let
  $\prog \eqdef \SB \overbrace{\{ a; b\} \hsep c}^{r_1};\;
  \overbrace{c \hsep 1 \leqslant \{ b = 1, \neg a = 1 \}}^{r_2};\;
  \overbrace{d \lor a \hsep }^{r_3} \SE$. Then, the sets~$\{a\}$,
  $\{c,d\}$ and $\{ b,c,d \}$ are answer sets of~$\prog$.
\end{example}%
Given a program~$\prog$, we consider the problems of computing an answer
set (called \AspComp) and outputting the number of optimal answer
sets (called \AspCountO).

\longversion{%
  Next, we show that under standard complexity-theoretic assumptions
  \AspCount is strictly harder than \#SAT.

  \begin{theorem}
    \AspCount for programs without optimization is $\cntc\co\NP$\hy
    complete.
  \end{theorem}
  \begin{proof}%
    Observe that programs containing choice and weight rules can be
    compiled to disjunctive ones (normalization) without these rule
    types (see~\cite{BomansonGebserJanhunen16}) using a polynomial
    number (in the original program size) of rules.  Membership
    follows from the fact that, given such a nice program $\prog$ and
    an interpretation $I$, checking whether $I$ is an answer of
    $\prog$ is $\co\NP$\hy complete, see
    e.g.,~\cite{KochL99}. Hardness is a direct consequence of
    $\cntc\co\NP$\hy hardness for the problem of counting subset
    minimal models of a CNF formula~\cite{DurandHK05}, since answer
    sets of negation-free programs and subset-minimal models of CNF
    formulas are essentially the same objects.
  \end{proof}

  \begin{remark}
    The counting complexity of \AspCount including optimization rules
    (i.e., where only optimal answer sets are counted) is slightly
    higher; exact results can be established employing hardness
    results from other sources~\cite{HermannP09}.
  \end{remark}
}%

\shortversion{\smallskip \noindent \textit{Tree Decompositions.}}%
\longversion{\subsection{Tree Decompositions}}%
\longversion{%

} %
Let $G = (V,E)$ be a graph, $T = (N,F,n)$ a rooted tree, and
$\chi: N \to 2^V$ a function that maps each node~$t \in N$ to a set of
vertices. We call the sets $\chi(\cdot)$ \emph{bags} and $N$ the set
of nodes. Then, the pair~${\mathcal{T}} = (T,\chi)$ is a \emph{tree
  decomposition (TD)} of~$G$ if the following conditions hold:
\begin{inparaenum}[(i)]
\item all vertices occur in some bag, that is, for every vertex~$v \in V$ there
  is a node~$t \in N$ with $v \in \chi(t)$;
\item all edges occur in some bag, that is, for every edge~$e \in E$ there is a
  node~$t \in N$ with $e \subseteq \chi(t)$; and
\item the \emph{connectedness condition}: for any three
  nodes~$t_1,t_2,t_3\in N$, if $t_2$ lies on the unique path
  from~$t_1$ to~$t_3$, then
  $\chi(t_1)\cap \chi(t_3) \subseteq \chi(t_2)$.
\end{inparaenum}
We call $\max\SB \Card{\chi(t)} - 1 \SM t \in N\SE$ the \emph{width}
of the TD. The \emph{treewidth}~$\tw{G}$ of a graph~$G$ is the
minimum width over all possible TDs of~$G$.
\longversion{%
}%
\longversion{Note that each graph has a trivial TD~$(T,\chi)$
  consisting of the tree~$(\{n\}, \emptyset, n)$ and the mapping
  $\chi(n) = V$.  It is well known that the treewidth of a tree
  is~$1$, and a graph containing a clique of size $k$ has at least
  treewidth $k-1$.}
For some arbitrary but fixed integer~$k$ and a graph of treewidth at
most~$k$, we can compute a TD of width $\leqslant k$ in
time~$2^{\bigO{k^3}} \cdot \Card{V}$~\cite{BodlaenderKoster08}.
Given a TD $(T,\chi)$ with $T = (N,\cdot,\cdot)$, for a node~$t \in N$
we say that $\type(t)$ is $\leaf$ if $t$ has no children; $\join$ if
$t$ has children~$t'$ and $t''$ with $t'\neq t''$ and
$\chi(t) = \chi(t') = \chi(t'')$; $\intr$ (``introduce'') if $t$ has a
single child~$t'$, $\chi(t') \subseteq \chi(t)$ and
$|\chi(t)| = |\chi(t')| + 1$; $\rem$ (``removal'') if $t$ has a single
child~$t'$, $\chi(t) \subseteq \chi(t')$ and
$|\chi(t')| = |\chi(t)| + 1$. If every node $t\in N$ has at most two
children, $\type(t) \in \{ \leaf, \join, \intr, \rem\}$, and bags of
leaf nodes and the root are empty, then the TD is called \emph{nice}.
For every TD, we can compute a nice TD in linear time without
increasing the width~\cite{BodlaenderKoster08}.
In our algorithms, we will traverse a TD bottom up, therefore, let
$\post(T,t)$ be the sequence of nodes in post-order of the induced
subtree~$T'=(N',\cdot, t)$ of $T$ rooted at~$t$. 

\begin{example}
  Figure~\ref{fig:graph-td} (left) shows a graph~$G_1$ together with a
  TD of~$G_1$ that is of width~$2$. Note that $G_1$ has treewidth~$2$,
  since it contains a clique on the vertices~$\{a,b,c\}$.  Further,
  the TD $\TTT$ in Figure~\ref{fig:running1_prim} is a nice TD of
  $G_1$.
\end{example}%
\shortversion{\smallskip \noindent \textit{Graph Representations of Programs.} }%
\longversion{\subsection{Graph Representations of Programs}}%
In order to use TDs for ASP solving, we need dedicated graph
representations of ASP programs.
The \emph{primal graph}~$P(\prog)$ of program~$\prog$ has the atoms of~$\prog$
as vertices and an edge~$a\,b$ if there exists a rule~$r \in \prog$ and $a,b
\in \at(r)$.
The \emph{incidence graph}~$I(\prog)$ of $\prog$ is the bipartite
graph that has the atoms and rules of~$\prog$ as vertices and an
edge~$a\, r$ if $a \in \at(r)$ for some rule~$r \in \prog$.
These definitions adapt similar concepts from
SAT~\cite{SamerSzeider10b}.

\begin{algorithm}[t]
  \KwData{Table algorithm ${\cal A}$, nice TD~$\TTT=(T,\chi)$ with
    $T=(N,\cdot,n)$ of $G(\prog)$ according to ${\cal A}$.}%
  \KwResult{Table: maps each TD node~$t\in T$ to some computed
    table~$\tau_t$. } %
  \For{\text{\normalfont iterate} $t$ in \text{\normalfont post-order}(T,n)}{\vspace{-0.05em}%
    $\Tab{} \eqdef \SB \Tabs{$t'$} \SM t' \text{ is a child of $t$ in
      $T$}\SE$\;\vspace{-0.05em} %
    $\Tabs{t} \eqdef {\cal A}(t,\chi(t),\prog_t,\atto,\Tab{})$\; %
    \vspace{-0.5em} }\vspace{-0.1em}%
  \caption{Algorithm ${\dpa}_{\cal A}({\cal T})$ for Dynamic Programming on TD ${\cal T}$ for ASP.}
\label{fig:dpontd}
\end{algorithm}
%

\begin{example}
  Recall program~$\prog$ of Example~\ref{ex:running1}. We observe that
  graph~$G_1$ ($G_2$) in the left (right) part of
  Figure~\ref{fig:graph-td} is the primal (incidence) graph
  of~$\prog$.
\end{example}

\shortversion{\smallskip \noindent \textit{Sub-Programs.} }%
\longversion{\subsection{Sub-Programs}}%
Let ${\cal T} = (T, \chi)$ be a nice TD of graph
representation~$H \in \{I(\prog), P(\prog)\}$ of a program
$\prog$. Further, let $T = (N,\cdot,n)$ and $t \in N$.
The \emph{bag-rules} are defined as $\prog_t \eqdef \SB r \SM r \in \prog,
\at(r)\subseteq {\chi(t)} \SE$ if $H$ is the primal graph and as $\prog_t \eqdef
\prog \cap \chi(t)$ if $H$ is the incidence graph.
Further, the set~$\atto \eqdef \SB a \SM a \in \at(\prog) \cap \chi(t'), t' \in
\post(T,t) \SE$ is called \emph{atoms below~$t$}, the \emph{program below $t$}
is defined as $\progt{t} \eqdef \SB r \SM r \in \prog_{t'}, t' \in \post(T,t) \SE$,
and the \emph{program strictly below $t$} is $\progtneq{t}\eqdef
\progt{t}\setminus \prog_t$. It holds that $\progt{n} = \progtneq{n} = \prog$ and
$\att{n}=\at(\prog)$. 
\begin{example}
  Intuitively, TDs of Figure~\ref{fig:graph-td} enable us to evaluate
  $\prog$ by analyzing sub-programs ($\{r_1, r_2\}$ and $\{r_3\}$) and
  combining results agreeing on $a$.  Indeed, for the given TD of
  Figure~\ref{fig:graph-td}~(left), $\progt{t_1}=\{r_1, r_2\}$,
  $\progt{t_2}=\{r_3\}$ and
  $\prog=\progt{t_3}=\progtneq{t_3}=\prog_{t_1} \cup \prog_{t_2}$.
  For the TD of Figure~\ref{fig:graph-td}~(right), we have
  $\progt{t_1} = \{r_1,r_2\}$ and $\att{t_1} = \{b,c\}$, as well as
  $\progt{t_3} = \{r_3\}$ and $\att{t_3} = \{a,d\}$.  Moreover, for TD
  ${\cal T}$ of Figure~\ref{fig:running1_prim},
  $\progt{t_1}\hspace{-0.1em}=\hspace{-0.1em}\progt{t_2}\hspace{-0.1em}=\hspace{-0.1em}\progt{t_3}\hspace{-0.1em}=\hspace{-0.1em}\progtneq{t_4}\hspace{-0.1em}=\hspace{-0.1em}\emptyset$,
  $\att{t_3}=\{a,b\}$ and $\progt{t_4}=\{r_1,r_2\}$.
\end{example}%
\section{ASP via Dynamic Programming on TDs}
\label{sec:algo:dp}
In the next two sections, we propose two dynamic programming (DP)
algorithms, $\dpa_{\PRIM}$ and $\dpa_{\INC}$, for ASP without
optimization rules based on two different graph representations, namely the primal
and the incidence graph. Both algorithms make use of the fact that
answer sets of a given program $\prog$ are (i)~models of $\prog$ and
(ii)~subset minimal with respect to~$\prog^M$. Intuitively, our
algorithms compute, for each TD node~$t$, (i)~sets of atoms---(local)
\emph{witnesses}---representing parts of potential models of~$\prog$,
and (ii)~for each local witness~$M$ subsets of~$M$---(local)
\emph{counterwitnesses}---representing subsets of potential models
of~$\prog^M$ which (locally) contradict that $M$ can be extended to an
answer set of~$\prog$.
We give the the basis of our algorithms in Algorithm~\ref{fig:dpontd}
($\dpa_{\cal A}$), which sketches the general DP scheme for ASP solving on TDs.
Roughly, the algorithm splits the search space based on a given nice TD and
evaluates the input program~$\prog$ in parts. The results are stored in
so-called tables, that is, sets of all possible tuples of witnesses and
counterwitnesses for a given TD node.
To this end, we define the \emph{table algorithms}~$\PRIM$ and $\INC$, which
compute tables for a node~$t$ of the TD using the primal graph~$P(\prog)$ and
incidence graph~$I(\prog)$, respectively.
To be more concrete, given a table algorithm~${\cal A} \in \{\PRIM, \INC\}$,
algorithm~$\dpa_{\cal A}$ visits every node~$t \in T$ in post-order; then, based
on~$\prog_t$, computes a table $\tab{t}$ for node $t$ from the tables of the
children of $t$, and stores $\tab{t}$ in~$\Tabs{t}$.
\footnoteitext{\label{foot:abrev}
  ${\cal S} \sqcup \{e\} \eqdef \SB S \cup \{e\} \SM S \in {\cal S}
  \SE$, $\MAIRR{S}{e} \eqdef S \cup \{e\}$, and
  $\MARRR{S}{e} \eqdef S \setminus \{e\}$}%
\subsection{Using Decompositions of Primal Graphs}\label{sec:prim}
%
%
 \begin{algorithm}[t]
   \KwData{Bag $\chi_t$, bag-rules $\prog_t$ and child tables $\Tab{}$ of node $t$.{~\bf Out:} Table~$\tab{t}$.} 
   %
   \lIf(\tcc*[f]{Abbreviations see
     Footnote~\ref{foot:abrev}.}){$\type(t) = \leaf$}{%
     $\tab{t} \eqdef \Big\{ \Big\langle
     \tuplecolor{\inputPredColor}{\emptyset},~\tuplecolor{\outputPredColor}{\emptyset}
     \Big\rangle \Big\}$%
     %
   }%
  \uElseIf{$\type(t) = \intr$, $a\in \chi_t$ is introduced and $\tau'\in \Tab{}$}{
   \vspace{-0.25em}
   \makebox[7.39cm][l]{$\tab{t} \eqdef \Big\{ \Big\langle \tuplecolor{\inputPredColor}{\MAI{M}},~\tuplecolor{\outputPredColor}{\Mod(\{M\} \cup [\CCC \sqcup \{ a \}] \cup \CCC, \prog_t^{\MAI{M}
     })} \Big\rangle$}
     $\Bigm|\;\langle \tuplecolor{\inputPredColor}{M}, \tuplecolor{\outputPredColor}{\CCC} \rangle \in \tab{}',  \tuplecolor{\inputPredColor}{\MAI{M}}
     \models {\prog}_t \Big\} \;\mcup$
   
     %
     \makebox[7.39cm][l]{\qquad$\,\,\,\hspace{0.07em}\Big\{ \Big\langle \tuplecolor{\inputPredColor}{M},~ \tuplecolor{\outputPredColor}{\Mod(\CCC,\prog_t^M)}\Big\rangle$}%
   $\Bigm|\;\langle \tuplecolor{\inputPredColor}{M}, \tuplecolor{\outputPredColor}{\CCC}
   \rangle \in \tab{}', \tuplecolor{\inputPredColor}{M} \models {\prog}_t \Big\}$
   \vspace{-0.05em}
       %
     }\vspace{-0.05em}%
     \uElseIf{$\type(t) = \rem$, $a \not\in \chi_t$ is removed and $\tau'\in \Tab{}$}{
       \makebox[7.39cm][l]{$\tab{t} \eqdef \Big\{ \Big\langle \tuplecolor{\inputPredColor}{\MAR{M}},~\tuplecolor{\outputPredColor}{\{ \MAR{C}
       \mid C \in \CCC \}}\Big\rangle$}$\Bigm|\;\langle \tuplecolor{\inputPredColor}{M}, \tuplecolor{\outputPredColor}{\CCC}
       \rangle \in \tab{}' \Big\}$
   \vspace{-0.3em}
     } %
     \uElseIf{$\type(t) = \join$ and $\tau', \tau'' \in \Tab{}$ with $\tau' \neq \tau''$}{%
       \makebox[7.39cm][l]{$\tab{t} \eqdef \Big\{ \Big\langle \tuplecolor{\inputPredColor}{M},~ \tuplecolor{\outputPredColor}{(\CCC' \cap \CCC'') \cup (\CCC' \cap \{M\}) \cup (\{M\}
       \cap \CCC'')}\Big\rangle$}$\Bigm|\;\langle \tuplecolor{\inputPredColor}{M}, \tuplecolor{\outputPredColor}{\CCC'}
       \rangle \in \tab{}', \langle \tuplecolor{\inputPredColor}{M}, \tuplecolor{\outputPredColor}{\CCC''} \rangle \in \tab{}''
       \Big\}$
   \vspace{-0.15em}
     } 
     \vspace{-0.15em}
 \caption{Table algorithm~$\algo{PRIM}(t,\chi_t,\prog_t,\cdot,\Tab{})$.}
   \label{fig:prim}
 \end{algorithm}%
In this section, we present our algorithm~\PRIM in two parts:
(i)~finding models of~$\prog$ and (ii)~finding models which are subset
minimal with respect to $\prog^M$. For sake of clarity, we first
present only the first tuple positions (red parts) of
Algorithm~\ref{fig:prim} (\PRIM) to solve 
(i). We call the
resulting table
algorithm~${\PRIMSAT}$.
\begin{example}\label{ex:sat}
  Consider program~$\prog$ from Example~\ref{ex:running1} and in
  Figure~\ref{fig:running1_prim} (left) TD~$\TTT=(\cdot, \chi)$
  of~$P(\prog)$ and the tables~$\tab{1}$, $\ldots$, $\tab{12}$, which
  illustrate computation results obtained during post-order traversal
  of ${\cal T}$ by $\dpa_{\PRIMSAT}$.
  Table~$\tab{1}=\SB \langle\emptyset\rangle \SE$ as
  $\type(t_1) = \leaf$.
  Since $\type(t_2) = \intr$, we construct table~$\tab{2}$
  from~$\tab{1}$ by taking~$M_{1.i}$ and $M_{1.i}\cup \{a\}$ for
  each~$M_{1.i} \in \tab{1}$ (corresponding to a guess on~$a$).  Then,
  $t_3$ introduces $b$ and $t_4$ introduces $c$.
  $\prog_{t_1}=\prog_{t_2}=\prog_{t_3} = \emptyset$, but since
  $\chi(t_4) \subseteq \at(r_1) \cup \at(r_2)$ we have
  $\prog_{t_4} = \{r_1, r_2\}$ for $t_4$.
  In consequence, for each~$M_{4.i}$ of table~$\tab{4}$, we have
  $M_{4.i} \models \{r_1, r_2\}$ since \PRIMSAT enforces
  satisfiability of $\prog_t$ in node~$t$.  We derive tables~$\tab{7}$
  to $\tab{9}$ similarly.  Since $\type(t_5) = \rem$, we remove
  atom~$b$ from all elements in $\tab{4}$ to construct $\tab{5}$. Note
  that we have already seen all rules where $b$ occurs and hence $b$
  can no longer affect witnesses during the remaining traversal. We similarly
  construct
  $\tab{t_6}=\tab{{10}}=\{\langle \emptyset \rangle, \langle a
  \rangle\}$.
  Since $\type(t_{11})=\join$, we construct table~$\tab{11}$ by taking
  the intersection $\tab{6} \cap \tab{{10}}$. Intuitively, this
  combines witnesses agreeing on~$a$.
  Node~$t_{12}$ is again of type~$\rem$.
  By definition (primal graph and TDs) for every~$r \in \prog$,
  atoms~$at(r)$ occur together in at least one common bag.
  Hence, $\prog=\progt{t_{12}}$ and since
  $\tab{12} = \{\langle \emptyset \rangle \}$, we can construct a
  model of~$\prog$ from the tables. For example, we obtain the
  model~$\{a,d\} = M_{11.2} \cup M_{4.2} \cup M_{9.3}$.
\end{example}%
\longversion{
  \begin{observation}\label{foot:clique_prim}
    Let $\prog$ be a program and $\cal T$ a TD of the primal graph of
    $\prog$. Then, for every rule~$r \in \prog$ there is at least one
    bag in $\cal T$ containing all atoms of~$r$.
  \end{observation}
  \begin{proof}
    By Definition the primal graph contains a clique on all atoms $a$
    participating in a rule $r$.  Since a TD must contain each edge of
    the original graph in some bag and has to be connected, it follows
    that there is at least one bag containing all (clique) atoms $a$
    of $r$.
  \end{proof}
}%
\begin{figure}[t]
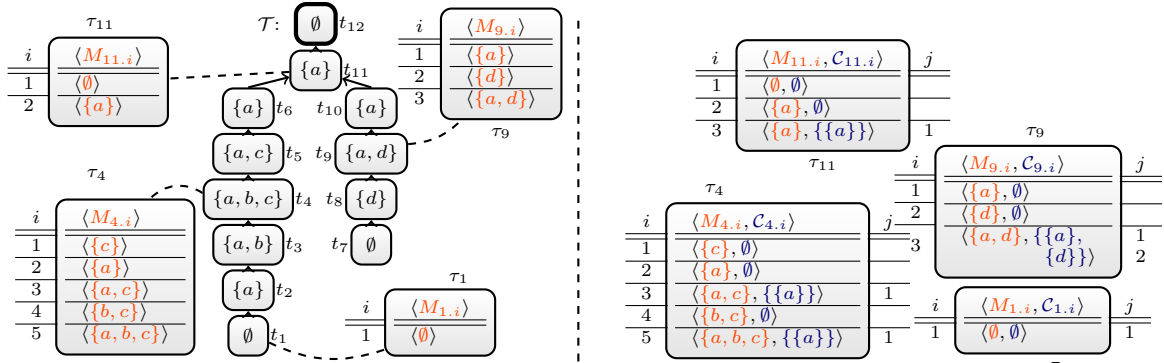
%
\centering %
\begin{minipage}{.50\linewidth}
  \begin{tikzpicture}[node distance=0.5mm]
\tikzset{every path/.style=thick}

\node (l1) [stdnode,label={[tdlabel, xshift=0em,yshift=+0em]right:${t_1}$}]{$\emptyset$};
\node (i1) [stdnode, above=of l1, label={[tdlabel, xshift=0em,yshift=+0em]right:${t_2}$}]{$\{a\}$};
\node (i12) [stdnode, above=of i1, label={[tdlabel, xshift=0em,yshift=+0em]right:${t_3}$}]{$\{a,b\}$};
\node (i13) [stdnode, above=of i12, label={[tdlabel, xshift=0em,yshift=+0em]right:${t_4}$}]{$\{a,b,c\}$};
\node (r1) [stdnode, above=of i13, label={[tdlabel, xshift=0em,yshift=+0em]right:${t_5}$}]{$\{a,c\}$};
\node (r12) [stdnode, above=of r1, label={[tdlabel, xshift=0em,yshift=+0em]right:${t_6}$}]{$\{a\}$};
\node (l2) [stdnode, right=2.5em of i12, label={[tdlabel, xshift=0em,yshift=+0em]left:${t_7}$}]{$\emptyset$};
\node (i2) [stdnode, above=of l2, label={[tdlabel, xshift=0em,yshift=+0em]left:${t_8}$}]{$\{d\}$};
\node (i22) [stdnode, above=of i2, label={[tdlabel, xshift=0em,yshift=+0em]left:${t_9}$}]{$\{a,d\}$};
\node (r2) [stdnode, above=of i22, label={[tdlabel, xshift=0em,yshift=+0em]left:${t_{10}}$}]{$\{a\}$};
\node (j) [stdnode, above left=of r2, yshift=-0.25em, label={[tdlabel, xshift=0em,yshift=+0em]right:${t_{11}}$}]{$\{a\}$};
\node (rt) [stdnode,ultra thick, above=of j, label={[tdlabel, xshift=0em,yshift=+0em]right:${t_{12}}$}]{$\emptyset$};
\node (label) [font=\scriptsize,left=of rt]{${\cal T}$:};
\node (leaf1) [stdnode, left=1em of i1, yshift=0.5em, label={[tdlabel, xshift=2em,yshift=+1em]above left:$\tab{4}$}]{\input{graph0_norm/tables/leaf0}};
\node (leaf1b) [stdnodenum,left=of leaf1,xshift=0.6em,yshift=0pt]{\input{graph0_norm/tables/leaf0num}};
\node (leaf2b) [stdnodenum,right=2.5em of j,xshift=-0.75em,yshift=+0.25em]  {\input{graph0_norm/tables/leaf1num}};
\node (leaf2) [stdnode,right=-0.4em of leaf2b, label={[tdlabel, xshift=0em,yshift=-0.25em]below:$\tab{9}$}]  {\input{graph0_norm/tables/leaf1}};
\coordinate (middle) at ($ (leaf1.north east)!.5!(leaf2.north west) $);
\node (join) [stdnode,left=2em of r12, yshift=1em, label={[tdlabel, xshift=2em,yshift=+0.25em]above left:$\tab{{11}}$}] {\input{graph0_norm/tables/join}};
\node (joinb) [stdnodenum,left=-0.45em of join] {\input{graph0_norm/tables/joinnum}};
\node (leaf0n) [stdnodenum,yshift=0.5em, right=2.5em of l1] {\input{graph0_norm/tables/leaf00num}};
\node (leaf0) [stdnode,right=-0.5em of leaf0n, label={[tdlabel, xshift=-1em,yshift=0.15em]above right:$\tab{1}$}] {\input{graph0_norm/tables/leaf00}};
\coordinate (top) at ($ (leaf2.north east)+(0.6em,-0.5em) $);
\coordinate (bot) at ($ (top)+(0,-12.9em) $);

\draw [->] (j) to (rt);
\draw [<-] (j) to ($ (r12.north)$);
\draw [<-] (j) to ($ (r2.north)$);
\draw [<-](r2) to (i22);
\draw [->](i2) to (i22);
\draw [->](l2) to (i2);
\draw [->](l1) to (i1);
\draw [<-](i12) to (i1);
\draw [<-](i13) to (i12);
\draw [<-](r1) to (i13);
\draw [<-](r12) to (r1);

\draw [dashed] (j) to (join);
\draw [dashed, bend right=20] (i22) to (leaf2);
\draw [dashed, bend right=40] (i13) to (leaf1);
\draw [dashed, bend left=22] (leaf0) to (l1);
\draw [dashed] (top) to (bot);
\end{tikzpicture}
%
%
\end{minipage}
\begin{minipage}{.48\linewidth}
\vspace{0.5em}
\begin{tikzpicture}[node distance=1mm]
\tikzset{every path/.style=thick}

\coordinate (start);
\node (leaf1) [stdnode,label={[tdlabel,  xshift=-2em,yshift=0.25em]above:$\tab{4}$}] {\input{graph0_norm/tables/leaf0_asp}};
\node (leaf1b) [stdnodenum,left=of leaf1,xshift=0.75em] {\input{graph0_norm/tables/leaf0_aspnum}};
\node (leaf1c) [stdnodenum,right=of leaf1,xshift=-0.75em] {\input{graph0_norm/tables/leaf0_aspnumb}};

\node (leaf0) [stdnode,right=1.3em of leaf1c, yshift=-1.5em, xshift=0em, label={[tdlabel,  yshift=-0.25em,xshift=1.0em,]below:$\tab{{1}}$}] {\input{graph0_norm/tables/leaf00_asp}};
\node (leaf0b) [stdnodenum,left=of leaf0,xshift=0.75em] {\input{graph0_norm/tables/leaf00_aspnum}};
\node (leaf0c) [stdnodenum,right=of leaf0,xshift=-0.75em] {\input{graph0_norm/tables/leaf00_aspnumb}};

\node (leaf2b) [stdnodenum,above=of leaf0b,xshift=-0.75em,yshift=0.85em]  {\input{graph0_norm/tables/leaf1_aspnum}};
\node (leaf2) [stdnode,right=of leaf2b, xshift=-0.75em, yshift=-0.45em, label={[tdlabel,  xshift=+0.1cm,yshift=0.25em]above:$\tab{9}$}]  {\input{graph0_norm/tables/leaf1_asp}};
\node (leaf2c) [stdnodenum,right=of leaf2, xshift=-0.75em]  {\input{graph0_norm/tables/leaf1_aspnumb}};

\coordinate (middle) at ($ (leaf1.north east)!.5!(leaf2.north west) $);
\node (join) [stdnode,above=2mm of middle,yshift=.25em, xshift=-3em, label={[tdlabel,  yshift=-0.25em]below:$\tab{{11}}$}] {\input{graph0_norm/tables/join_asp}};
\node (joinb) [stdnodenum,left=of join,xshift=0.75em] {\input{graph0_norm/tables/join_aspnum}};
\node (joinc) [stdnodenum,right=of join,xshift=-0.75em] {\input{graph0_norm/tables/join_aspnumb}};
%
%
%
\end{tikzpicture}
%
\vspace{-1.2em}
\end{minipage}
\caption{Selected DP tables of~\PRIMSAT~(left) and~\PRIM~(right) for
  nice TD~$\TTT$.}
\label{fig:running1_prim}
\label{fig:running1_prim_asp}
\end{figure}%
\PRIM is given in Algorithm~\ref{fig:prim}.  Tuples in~$\tab{t}$ are
of the form~$\langle M, \CCC \rangle$.  Witness~$M \subseteq \chi(t)$
represents a model of~$\prog_t$ witnessing the existence of
$M' \supseteq M$ with $M' \models \progt{t}$.  The
family~$\CCC \subseteq 2^M$ contains sets of models~$C \subseteq M$ of
the GL reduct~$(\prog_t)^M$. $C$ witnesses the existence of a set~$C'$
with counterwitness~$C \subseteq C'\subsetneq M'$ and
$C' \models (\progt{t})^{M'}$.
There is an answer set of~$\prog$ if table~$t_n$ for root~$n$ contains
$\langle \emptyset, \emptyset \rangle$.
Since in Example~\ref{ex:sat} we already explained the first tuple position 
and thus the witness part, we only briefly describe the parts
for counterwitnesses.
In the introduce case, we want to store only counterwitnesses for not
being minimal with respect to the GL reduct of the bag-rules.
Therefore, in Line~3 we construct for $\MAI{M}$ counterwitnesses from
either some witness $M$ ($M \subsetneq \MAI{M}$), or of any
$C \in \CCC$, or of any $C \in \CCC$ extended by~$a$
(every~$C \in\CCC$ was already a counterwitness before).
Line~4 ensures that only counterwitnesses that are models of the GL
reduct $\prog_t^M$ are stored (via $\Mod(\cdot, \cdot)$).
Line~6 restricts counterwitnesses to its bag content, and Line~8
enforces that child tuples agree on counterwitnesses.

\begin{example}\label{ex:prim:min}
  Consider Example~\ref{ex:sat}, its TD~$\TTT=(\cdot,\chi)$,
  Figure~\ref{fig:running1_prim} (right), and the tables~$\tab{1}$,
  $\ldots$, $\tab{12}$ obtained by $\dpa_{\PRIM}$.
  Since we have $\at(r_1) \cup \at(r_2) \subseteq \chi(t_4)$, we
  require $C_{4.i.j} \models \{r_1, r_2\}^{M_{4.i}}$ for each
  counterwitness~$C_{4.i.j} \in \CCC_{4.i}$ in tuples of~$\tab{4}$.
  For $M_{4.5} = \{a,b,c\}$ observe that the only counterwitness
  of
  $\{r_1, r_2\}^{M_{4.5}} = \{a \leftarrow c, b \leftarrow c, c
  \leftarrow 1 \leq \{b = 1\}\}$ is $C_{4.5.1} = \{a\}$.
  Note that witness $M_{11.2}$ of table $\tab{11}$ is the result of
  joining $M_{4.2}$ with $M_{9.1}$ and witness $M_{11.3}$
  (counterwitness $C_{11.3.1}$) is the result of joining $M_{4.3}$
  with $M_{9.3}$ ($C_{4.3.1}$ with $C_{9.3.1}$), and $M_{4.5}$ with
  $M_{9.3}$ ($C_{4.5.1}$ with $C_{9.3.2}$).  $C_{11.3.1}$ witnesses
  that neither $M_{4.3} \cup M_{9.3}$ nor $M_{4.5} \cup M_{9.3}$ forms
  an answer set of~$\prog$.
  Since $\tab{12}$ contains $\langle \emptyset, \emptyset \rangle$
  there is no counterwitness for $M_{11.2}$, we can construct an
  answer set of~$\prog$ from the tables,~e.g., $\{a\}$ can be
  constructed from $M_{4.2} \cup M_{9.1}$.
\end{example}%
\begin{theorem}\label{thm:prim:runtime}
  Given a program~$\prog$, the algorithm ${\dpa}_{\PRIM}$ is correct
  and runs in time~$\bigO{2^{2^{k+2}} \cdot \CCard{P(\prog)}}$ where
  $k$~is the treewidth of the primal graph~$P(\prog)$.
  \vspace{-0.5em}
\end{theorem}
\shortversion{%
  We omit the proof due to space reasons and refer to
  Appendix~\ref{proof:thm:prim:runtime}.
}%
\longversion{%
  \begin{proof}
    We refer to Appendix~\ref{proof:thm:prim:runtime}.
  \end{proof}
}

%
\begin{algorithm}[t]
   \KwData{Bag $\chi_t$, bag-rules $\prog_t$, atoms-below $\atto$, child tables $\Tab{}$ of $t$.{~\bf Out:} Tab.~$\tab{t}$.} 
   %

   %
   %
   \lIf(\tcc*[f]{\hspace{-0.05em}Abbreviations see Footnote~\ref{foot:sigma}.\hspace{-0.05em}}){$\type(t) = \leaf$}{%
     $\tab{t} \eqdef \Big\{ \Big\langle
     \tuplecolor{\inputPredColor}{\emptyset},
     \tuplecolor{\statePredColor}{\emptyfunc},
     ~\tuplecolor{\outputPredColor}{\emptyset} \Big\rangle \Big\}$
   \vspace{-0.00em}}%
   %
   %
   %
   \uElseIf{$\type(t) = \intr$, $a \in \chi_t \setminus \prog_t$ is introduced and $\tau'\in \Tab{}$}{     
     \vspace{-0.25em}$\tab{t} \eqdef \Big\{ \Big\langle \tuplecolor{\inputPredColor}{\MAI{M}}, \tuplecolor{\statePredColor}{\sigma \squplus
     \SSR(\dot\prog_t^{(t,\sigma)},\MAI{M}) },~\tuplecolor{\outputPredColor}{\{\langle M,
       \sigma \squplus \SSR(\dot\prog_t^{(t,\sigma,{\MAI{M}})},M)\rangle \}~\cup}$

  \vspace{-0.2em}%
       \makebox[1.5cm][l]{}\makebox[7.05cm][l]{$\tuplecolor{\outputPredColor}{\{ \langle\MAI{C},
       \rho \squplus \SSR(\dot\prog_t^{(t,\rho,{\MAI{M}})},\MAI{C})\rangle \mid \langle C,
       \rho\rangle \in \CCC \}~\cup}$}

\vspace{-0.2em}%
       \makebox[1.5cm][l]{}\makebox[7.05cm][l]{$\tuplecolor{\outputPredColor}{\{ \langle C,
       \rho \squplus \SSR(\dot\prog_t^{(t,\rho,{\MAI{M}})},C)\rangle \mid \langle C, \rho
       \rangle \in \CCC \}}\Big\rangle$} 
	$\Bigm|\;\langle \tuplecolor{\inputPredColor}{M}, \tuplecolor{\statePredColor}{\sigma}, \tuplecolor{\outputPredColor}{\CCC} \rangle \in \tab{}'\Big\}\;\mcup$

	\qquad\,\,\,$\Big\{ \Big\langle \tuplecolor{\inputPredColor}{M}, \tuplecolor{\statePredColor}{\sigma \squplus
     \SSR(\dot\prog_t^{(t,\sigma)},M)},$%

	\vspace{-0.15em}%
     \makebox[1.5cm][l]{}\makebox[7.05cm][l]{$\tuplecolor{\outputPredColor}{\{ \langle C,
       \rho \squplus \SSR(\dot\prog_t^{(t,\rho,M)},C) \rangle \mid \langle C, \rho\rangle \in \CCC
       \}}\Big\rangle$}
     $\Bigm|\;\langle \tuplecolor{\inputPredColor}{M}, \tuplecolor{\statePredColor}{\sigma}, \tuplecolor{\outputPredColor}{\CCC} \rangle \in \tab{}'
     \Big\}$

     %
   \vspace{-0.10em}}
   %
   %
   \uElseIf{$\type(t) = \intr$, $r \in \chi_t \cap \prog_t$ is introduced and $\tau'\in \Tab{}$}{
     $\tab{t} \eqdef\Big\{ \Big\langle \tuplecolor{\inputPredColor}{M}, \tuplecolor{\statePredColor}{\MAIR{\sigma}{r} \squplus
     \SSR(\{\dot r\}^{(t,\MAIR{\sigma}{r}\})},M)},$
 
 	\vspace{-0.15em}%
     \makebox[1.5cm][l]{}\makebox[7.55cm][l]{$\tuplecolor{\outputPredColor}{\{ \langle C, \MAIR{\rho}{r} \squplus
       \SSR(\{\dot r\}^{(t,\MAIR{\rho}{r},M)},C) \rangle \mid \langle C, \rho \rangle \in \CCC\}} \Big\rangle$}
     $\Bigm|\;\langle \tuplecolor{\inputPredColor}{M}, \tuplecolor{\statePredColor}{\sigma}, \tuplecolor{\outputPredColor}{\CCC} \rangle \in \tab{}' \Big\}$
   \vspace{-0.10em}}
   %
   %
   %
   \uElseIf{$\type(t) = \rem$, $a \not\in \chi_t$ is removed atom and $\tau'\in \Tab{}$}{
       $\tab{t} \eqdef\Big\{ \Big\langle \tuplecolor{\inputPredColor}{\MAR{M}},\tuplecolor{\statePredColor}{\sigma \squplus
       \UpdateStates(\prog_t,M,a)},$
	   
	   \makebox[1.5cm][l]{}\makebox[7.55cm][l]{$\tuplecolor{\outputPredColor}{\{ \langle \MAR{C}, \rho \squplus \UpdateRedStates(\prog_t,M,C,a) \rangle
       \mid \langle C, \rho \rangle \in \CCC \}} \Big\rangle$}
       $\Bigm|\;\langle \tuplecolor{\inputPredColor}{M}, \tuplecolor{\statePredColor}{\sigma}, \tuplecolor{\outputPredColor}{\CCC} \rangle \in \tab{}'\Big\}$
       %
%
     %
   \vspace{-0.07em}}%
   %
   %
   \uElseIf{$\type(t) = \rem$, $r \not\in \chi_t$ is removed rule and $\tau'\in \Tab{}$}{
     $\tab{t} \eqdef$ \makebox[6.645cm][l]{$\Big\{\Big\langle \tuplecolor{\inputPredColor}{M},
       \tuplecolor{\statePredColor}{\MARR{\sigma}{\hspace{-0.1em}\{r\}}},$~%
       $\tuplecolor{\outputPredColor}{\big\{ \langle C, \MARR{\rho}{\hspace{-0.1em}\{r\}} \rangle \mid \langle C, \rho
       \rangle \in \CCC, \rho(r) =  
       \infty \big\}} \Big\rangle$}%
%
     $\Bigm|\;\langle \tuplecolor{\inputPredColor}{M}, \tuplecolor{\statePredColor}{\sigma}, \tuplecolor{\outputPredColor}{\CCC} \rangle \in \tab{}', \sigma 
     (r) =  \infty\hspace{-0.015em} \Big\}$ 
     \vspace{-0.02em}}

   %
   %
   \uElseIf{$\type(t) = \join$ and $\tau', \tau'' \in \Tab{}$ with $\tau' \neq \tau''$}{
     $\tab{t} \eqdef\Big\{ \Big\langle \tuplecolor{\inputPredColor}{M}, \tuplecolor{\statePredColor}{\sigma' \squplus \sigma''},~\tuplecolor{\outputPredColor}{\{
     \langle C, \rho' \squplus \rho''\rangle \mid \langle C, \rho'
     \rangle \in \CCC', \langle C, \rho'' \rangle \in \CCC''\}~\cup}$%
  
     \makebox[1.5cm][l]{}$\tuplecolor{\outputPredColor}{\{ \langle M, \rho \squplus \sigma''\rangle
     \mid \langle M, \rho \rangle \in \CCC' \}~\cup}$%
     
     \makebox[1.5cm][l]{}\makebox[4.7cm][l]{$\tuplecolor{\outputPredColor}{\{ \langle M, \sigma'
       \squplus \rho \rangle \mid \langle M, \rho \rangle \in \CCC'' \}}%
       \Big\rangle$}%
       \,\,$\Bigm|\;\langle \tuplecolor{\inputPredColor}{M}, \tuplecolor{\statePredColor}{\sigma'}, \tuplecolor{\outputPredColor}{\CCC'} \rangle \in \tab{}',
       \langle \tuplecolor{\inputPredColor}{M}, \tuplecolor{\statePredColor}{\sigma''}, \tuplecolor{\outputPredColor}{\CCC''} \rangle \in \tab{}'' \Big\}$ 
     \vspace{-0.15em}%
   }
   \vspace{-0.05em}
\caption{Table algorithm~$\algo{INC}(t,\chi_t,\prog_t,\atto,\Tab{})$.}
\label{fig:incinc}
\end{algorithm}


%
\subsection{Using Decompositions of Incidence Graphs}\label{sec:inc}
Our next algorithm ($\dpa_{\INC}$) takes the incidence graph as graph
representation of the input program.  The treewidth of the incidence
graph is smaller than the treewidth of the primal graph plus one,
cf.,~\cite{SamerSzeider10b,FichteSzeider15}.  More importantly, the
incidence graph does not enforce cliques on $\at(r)$ for some rule~$r$.
The incidence graph, compared to the primal graph, additionally
contains rules as vertices and its relationship to the atoms in terms
of edges. By definition, we have no guarantee that all atoms of a rule
occur together in the same bag of TDs of the incidence graph.
For that reason, we \emph{cannot} locally check the satisfiability of a rule
when traversing the TD without additional stored information (so-called
\emph{rule-states} that intuitively represent how much of a rule is already
(dis-)satisfied).
We only know that for each rule~$r$ there is a
path~$p=t_\intr,t_1,\ldots,t_m,t_\rem$ where $t_\intr$ introduces~$r$
and $t_\rem$ removes~$r$ and when considering $t_\rem$ in the table
algorithm we have seen all atoms that occur in
rule~$r$.
%
%
Thus, on removal of~$r$ in $t_\rem$ we ensure that $r$ is satisfied
while taking rule-states for choice and weight rules into account.
Consequently, our tuples will contain a witness, its rule-state, 
and counterwitnesses and their rule-states.

\shortversion{%
\footnoteitext{\label{foot:sigma}\label{foot:abrevtwo}%
  $\sigma \squplus \rho\hspace{-0.15em}\eqdef\hspace{-0.1em}\SB
  (x,\hspace{-1.0em}{\underset{(x,c_1)\in\sigma}{\hspace{-0.2em}\Sigma}}\hspace{-1.2em}c_1
  \hspace{0.05em}+\hspace{-0.95em} 
  {\underset{\hspace{0.8em}(x,c_2)\in\rho}{\hspace{-0.5em}\Sigma}\hspace{-1.7em}c_2}\hspace{0.0em})
  \SM (x,\cdot)\in \sigma\cup\rho$ $\hspace{-0.3em}
  \SE$;
  $\MAIR{\sigma}{r}\hspace{-0.15em}\eqdef\hspace{-0.1em}\sigma \cup
  \{(r,0)\}$;
  $\MARR{\sigma}{S}\hspace{-0.15em}\eqdef\hspace{-0.1em}\{(x,y) \in
  \sigma \mid x\not\in S\}$.
}
} \longversion{%
  \footnoteitext{\label{foot:sigma}\label{foot:abrevtwo}%
    $\sigma \squplus \rho\eqdef\SB (x,{\Sigma_{(x,c_1) \in\sigma}}c_1+
    {\Sigma}_{(x,c_2)\in\rho}c_2) \SM (x,\cdot)\in \sigma\cup\rho
    \SE$;
    $\MAIR{\sigma}{r}\eqdef\sigma \cup \{(r,0)\}$;
    $\MARR{\sigma}{S}\eqdef\{(x,y) \in \sigma \mid x\not\in
    S\}$.
}
}

\newcommand{\sat}{\text{sat}}
\newcommand{\undec}{\text{un}}
\newcommand{\SatRules}{SR}
A tuple in
$\tab{t}$ for Algorithm~\ref{fig:incinc} (\INC) is a triple~$\langle
M, \sigma, \CCC \rangle$. The set~$M \subseteq
\at(\prog)\cap\chi(t)$ represents again a witness. A
\emph{rule-state}~$\sigma$ is a mapping~$\sigma: \prog_t \rightarrow
\NAT_0 \cup \{\infty\}$. A rule state for
$M$ represents whether rules of
$\chi(t)$ are either (i)~satisfied by a superset of
$M$ or (ii)~undecided for~$M$.
Formally, the set~$\SatRules(\prog_t,\sigma)$ of satisfied bag-rules~$\prog_t$
consists of each rule~$r \in \prog_t$ such that $\sigma(r) = \infty$.
Hence, $M$ witnesses a model~$M'\supseteq M$ where $M' \models
\progtneq{t} \cup \SatRules(\prog_t,\sigma)$.
$\CCC$ concerns counterwitnesses.

We compute a new
rule-state~$\sigma$ from a rule-state, ``updated'' bounds for weight
rules ($\UpdateStates$), and satisfied rules
($\SSR$, defined below).  We define $\UpdateStates(\prog_t, M, a)
\eqdef \sigma'$ depending on an atom~$a$ with $\sigma'(r) \eqdef
\wght(r,\{a\} \cap [(B^-_r \setminus M) \cup (B^+_r \cap M)])$, if $r
\in \weight(\prog_t)$.
We use binary
operator~$\squplus^{\ref{foot:abrevtwo}}$ to combine rule-states,
which ensures that rules satisfied in at least one operand remain
satisfied. Next, we explain the meaning of rule-states.

\begin{example}\label{ex:rulestates}%
  Consider program~$\prog$ from Example~\ref{ex:running1} and
  TD~$\TTT'=(\cdot,\chi)$ of~$I(\prog)$ and the tables~$\tab{1}$,
  $\ldots$, $\tab{18}$ in Figure~\ref{fig:running1_inc} (left).
  We are only interested in the first two tuple positions (red and
  green parts) and implicitly assume that ``$i$'' refers to Line~$i$
  in the respective table.  Consider $M_{4.1}=\{c\}$ in
  table~$\tab{4}$. Since $H_{r_2} = \{c\}$, witness $M_{4.1}=\{c\}$
  satisfies rule~$r_2$.
  As a result, $\sigma_{4.1}(r_2)=\infty$ remembering satisfied
  rule~$r_2$ for~$M_{4.1}$. Since $c \notin M_{4.2}$ and
  $B^+_{r_1}=\{c\}$, $M_{4.2}$ satisfies rule~$r_1$, resulting in
  $\sigma_{4.2}(r_1)=\infty$. Rule-state~$\sigma_{4.1}(r_1)$
  represents that $r_1$ is undecided for $M_{4.2}$.
  For weight rule~$r_2$, rule-states remember the sum of body weights
  involving removed atoms.  Consider $M_{6.2}=M_{6.3}=\emptyset$ of
  table~$\tab{6}$. We have $\sigma_{6.2}(r_2)\neq\sigma_{6.3}(r_2)$,
  because $M_{6.2}$ was obtained from some~$M_{5.{i}}$ of
  table~$\tab{5}$ with $b\not\in M_{5.{i}}$ and $b$ occurs in
  $B^+_{r_2}$ with weight~$1$, resulting in $\sigma_{6.3}(r_2)=1$;
  whereas $M_{6.3}$ extends some~$M_{5.{j}}$ with $b \notin M_{5.j}$.
\end{example}%
\noindent In order to decide in node~$t$ whether a witness satisfies
rule~$r\in\prog_t$, we check satisfiability of
program~$\dot{\cal R}(r)$ constructed by~$\dot{\cal R}$, which maps
rules to state-programs.  Formally, for
$M \subseteq \chi(t) \setminus \prog_t$,
$\SSR(\dot{\cal R},M)\eqdef \sigma$ where $\sigma(r) \eqdef \infty$ if
$(r,{\cal R})\in\dot{\cal R}$ and $M \models {\cal R}$.
\begin{definition}\label{def:bagprogram}%
  Let $\prog$ be a program, $\TTT=(\cdot,\chi)$ be a TD of $I(\prog)$,
  $t$ be a node of $\TTT$, ${\cal P} \subseteq \prog_t$, and
  $\sigma: \prog_t \rightarrow \NAT_0 \cup \{\infty\}$ be a
  rule-state.
  The \emph{state-program}~${\cal P}^{(t,\sigma)}$ is obtained
  from~${\cal P} \cup \{ \leftarrow B_r \mid r \in \choice({\cal P}),
  H_r \subsetneq \att{t}\}$\shortversion{\footnote{%
      We require to add
      $\{\leftarrow B_r \mid r\in \choice({\cal P}), H_r \subsetneq
      \att{t}\}$ in order to decide satisfiability for corner cases of
      choice rules involving counterwitnesses of Line~3 in
      Algorithm~\ref{fig:incinc}.}}%
  \longversion{\footnote{%
      We require to add
      $\{\leftarrow B_r \mid r\in \choice({\cal P}), H_r \subsetneq
      \att{t}\}$ in order to decide satisfiability for corner cases of
      choice rules
      involving\newline \hbox{~}\hspace{1.2em} counterwitnesses of Line~3 in
      Algorithm~\ref{fig:incinc}.}} %
  by \shortversion{\vspace{-1.25ex}}
  \begin{enumerate}
  \item \shortversion{\hspace{-0.2em}}removing rules~$r$ with~$\sigma(r)=\infty$
    (``already satisfied rules'');
  \item \shortversion{\hspace{-0.20em}}removing from every rule all
    literals~$a, \neg a$ with $a \not\in \chi(t)$; and
  \item \shortversion{\hspace{-0.50em}} setting new bound
    $\max\{ 0, \bnd(r)\hspace{-0.0em} - \sigma(r)
    \hspace{-0.0em}\hspace{-0.05em}- \hspace{-0.05em} \wght(r, \at(r)
    \setminus \atto)\}$ for weight
    rule~$r$.
  \end{enumerate}
  \shortversion{\vspace{-0.5em} }%
  We
  define~$\dot{\cal P}^{(t,\sigma)}: {\cal P} \rightarrow 2^{{\cal
      P}^{(t,\sigma)}}$ by
  $\dot{\cal P}^{(t,\sigma)}(r) \eqdef \{r\}^{(t,\sigma)}$ for
  $r\in {\cal P}$.
\end{definition}%
\begin{example}
  Observe
  $\prog_{t_1}^{(t_1,\emptyfunc)} = \{\{b\} \leftarrow c, \leftarrow
  c, c \leftarrow 0 \leq \{ b = 1 \} \}$ and
  $\prog_{t_2}^{(t_2,\emptyfunc)} = \{\{a\} \leftarrow, \leftarrow 1
  \leq \{ \neg a = 1 \} \}$ for $\prog_{t_1}$, $\prog_{t_2}$ of
  Figure~\ref{fig:graph-td}~\hspace{-0.3em}(right).
\end{example}
\longversion{%

}%
The following example provides an idea how we compute models of a
given program using the incidence graph.  The resulting algorithm
\INCSAT is the same as \INC, except that only the first two tuple positions (red
and green parts) are considered.
\begin{example}
  Again, we consider~$\prog$ of Example~\ref{ex:running1} and in
  Figure~\ref{fig:running1_inc}~(left) ${\cal T}'$ as well as
  tables~$\tab{1}$, $\ldots$, $\tab{18}$.  Table
  $\tab{1}= \{\langle \emptyset, \emptyfunc \rangle \}$ as
  $\type(t_1)=\leaf$.  Since $\type(t_2)=\intr$ and $t_2$ introduces
  atom~$c$, we construct $\tab{2}$ from $\tab{1}$ by taking
  $M_{2.1}\eqdef M_{1.1}\cup \{c\}$ and $M_{2.2} \eqdef M_{1.1}$ as
  well as rule-state $\emptyfunc$.  
  Because $\type(t_3)=\intr$ and $t_3$ introduces rule~$r_1$, we
  consider state program $L_3\eqdef \{r_1\}^{(t_3,\{(r_1,0)\})}=$
  \hbox{$\SB \hsep c \SE$} for $\SSR(\dot L_3,M_{2.1})=\{(r_1,0)\}$ as
  well as $\SSR(\dot L_3,M_{2.2})=\{(r_1,\infty)\}$ (according to
  Line~9 of Algorithm~\ref{fig:incinc}).
  Because $\type(t_4)=\intr$ and $t_4$ introduces rule~$r_2$, we
  consider $M_{3.1} \eqdef M_{2.1}$ and $M_{3.2} \eqdef M_{2.2}$ and
  state program
  $L_4\eqdef \{r_2\}^{(t_4,\{(r_2,0)\})}=\SB c \hsep 0 \leqslant \{\}
  \SE =$ \hbox{$\SB c \hsep \SE$} for
  $\SSR(\dot L_4,M_{3.1})=\{(r_2,\infty)\}$ as well as
  $\SSR(\dot L_4,M_{3.2})=\{(r_2,0)\}$ (see Line~9).
  Node~$t_5$ introduces~$b$ (table not shown) and node~$t_6$
  removes~$b$. Table~$\tab{6}$ was discussed in
  Example~\ref{ex:rulestates}.  When we remove~$b$ in $t_6$ we have
  decided the ``influence'' of $b$ on the satisfiability of $r_1$ and
  $r_2$ and thus all rules where $b$ occurs.
  Tables~$\tab{7}$ and~$\tab{8}$ can be derived similarly. Then, $t_9$
  removes rule~$r_2$ and we ensure that every witness~$M_{9.1}$ can be
  extended to a model of~$r_2$, i.e., witness candidates for $\tab{9}$
  are $M_{8.i}$ with~$\sigma_{8.i}(r_2)=\infty$. The remaining tables
  are derived similarly. For example, table~$\tab{17}$ for join
  node~$t_{17}$ is derived analogously to table $\tab{17}$ for
  algorithm~\PRIM in Figure~\ref{fig:running1_prim}, but, in addition,
  also combines the rule-states as specified in
  Algorithm~\ref{fig:incinc}.
\end{example}
\begin{figure}[t]
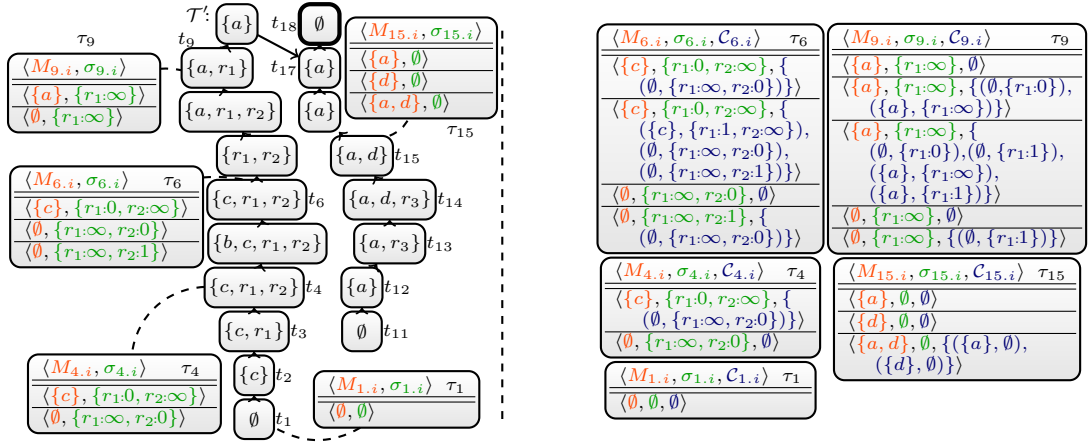
%
  \centering %
  \begin{minipage}{.48\linewidth}
    \hspace{0.2em}%
    \begin{tikzpicture}[node distance=0.5mm]
\tikzset{every path/.style=thick}

\node (l1) [stdnodecompact, label={[tdlabel, xshift=0em,yshift=+0em]right:${t_1}$}]{$\emptyset$};
\node (i1) [stdnodecompact, above=of l1, label={[tdlabel, xshift=0em,yshift=+0em]right:${t_2}$}]{$\{c\}$};
\node (i12) [stdnodecompact, above=of i1, label={[tdlabel, xshift=0em,yshift=+0em]right:${t_3}$}]{$\{c,r_1\}$};
\node (i13) [stdnodecompact, above=of i12, label={[tdlabel, xshift=0em,yshift=+0em]right:${t_4}$}]{$\{c,r_1,r_2\}$};
\node (r1) [stdnodecompact, above=of i13, xshift=0.5em, label={[tdlabel, xshift=0em,yshift=+0em]right:$ $}]{$\{b,c,r_1,r_2\}$};
\node (r12) [stdnodecompact, above=of r1, xshift=-0.4em, label={[tdlabel, xshift=0em,yshift=+0em]right:${t_6}$}]{$\{c,r_1,r_2\}$};
\node (r13) [stdnodecompact, above=of r12, label={[tdlabel, xshift=0em,yshift=+0em]right:$ $}]{$\{r_1,r_2\}$};
\node (i14) [stdnodecompact, above=of r13, xshift=-1em, label={[tdlabel, xshift=0em,yshift=+0.75em]left:$ $}]{$\{a,r_1,r_2\}$};
\node (r14) [stdnodecompact, above=of i14, xshift=-0.5em, label={[tdlabel, xshift=0em,yshift=+0.0em]above left:${t_9}$}]{$\{a,r_1\}$};
\node (r15) [stdnodecompact, above=of r14,xshift=0.75em, label={[tdlabel, xshift=0em,yshift=+0em]right:$ $}]{$\{a\}$};
\node (l2) [stdnodecompact, right=2.0em of i12, label={[tdlabel, xshift=0em,yshift=+0em]right:$t_{11}$}]{$\emptyset$};
\node (i2) [stdnodecompact, above=of l2, label={[tdlabel, xshift=0em,yshift=+0em]right:$t_{12}$}]{$\{a\}$};
\node (i22) [stdnodecompact, above=of i2, xshift=1.05em, label={[tdlabel, xshift=0em,yshift=+0em]right:$t_{13}$}]{$\{a,r_3\}$};
\node (i23) [stdnodecompact, above=of i22, xshift=0.0em, label={[tdlabel, xshift=0em,yshift=+0em]right:$t_{14}$}]{$\{a,d,r_3\}$};
\node (r2) [stdnodecompact, above=of i23, xshift=-1.05em, label={[tdlabel, xshift=0em,yshift=+0em]right:${t_{15}}$}]{$\{a,d\}$};
\node (r22) [stdnodecompact, above=of r2, xshift=-1.6em, label={[tdlabel, xshift=0em,yshift=+0em]right:$ $}]{$\{a\}$};
\node (j) [stdnodecompact, above=of r22, label={[tdlabel, xshift=0em,yshift=+0em]left:$t_{17}$}]{$\{a\}$};
\node (rt) [stdnodecompact, ultra thick,above=of j, label={[tdlabel, xshift=0em,yshift=+0em]left:$t_{18}$}]{$\emptyset$};
\node (label) [font=\scriptsize,left=of r15,yshift=0.40em,xshift=0.3em]{${\cal T}'$\hspace{-0.25em}:};
\coordinate (r12pos) at ($ (r12.north)+(-0.5em,0) $);
\node (leaf1) [stdnodetable, left=0.5em of i1, yshift=-0.7em, label={[tdlabel, xshift=0em,yshift=+0.25em]above left:$ $}]{\input{graph0_norm/tables_inc/leaf0}};
\node (leaf2) [stdnodetable,right=0.2em of j, label={[tdlabel, xshift=1.5em,yshift=-0.25em]below:$\tab{15}$}]  {\input{graph0_norm/tables_inc/leaf1}};
\coordinate (middle) at ($ (leaf1.north east)!.5!(leaf2.north west) $);
\node (join) [stdnodetable,left=0.15em of r12, yshift=-0.65em, label={[tdlabel, xshift=2em,yshift=-0.25em]below left:$ $}] {\input{graph0_norm/tables_inc/join}};
\node (r14n) [stdnodetable,left=0.80em of r14, yshift=-0.95em, label={[tdlabel, xshift=2em,yshift=+0.25em]above left:$\tab{{9}}$}] {\input{graph0_norm/tables_inc/r14}};
\node (leaf0) [stdnodetable,right=1.5em of l1,yshift=0.75em, label={[tdlabel, xshift=-0.0em,yshift=0.15em]above right:$ $}] {\input{graph0_norm/tables_inc/leaf00}};
\coordinate (top) at ($ (leaf2.north east)+(0.2em,-0.5em) $);
\coordinate (bot) at ($ (top)+(0,-14.8em) $);

\draw [->] (j) to (rt);
\draw [->] (r15) to (j);
\draw [->](i2) to (i22);
\draw [->](r2) to (r22);
\draw [->](r22) to (j);
\draw [->](i22) to (i23);
\draw [->](i23) to (r2);
\draw [->](l2) to (i2);
\draw [->](l1) to (i1);
\draw [<-](i12) to (i1);
\draw [<-](i13) to (i12);
\draw [<-](r1) to (i13);
\draw [->](r12) to (r13);
\draw [->](r13) to (i14);
\draw [->](i14) to (r14);
\draw [->](r14) to (r15);
\draw [<-](r12) to (r1);

\draw [dashed,bend left=7] (join) to (r12pos);
\draw [dashed,bend left=7] (r14n) to (r14);
\draw [dashed, bend right=20] (r2) to (leaf2);
\draw [dashed, bend right=40] (i13) to (leaf1);
\draw [dashed, bend left=30] (leaf0) to (l1);
\draw [dashed] (top) to (bot);
\end{tikzpicture}
%
%
  \end{minipage}
\begin{minipage}{.51\linewidth}
  \hspace{0.05em}%
  \begin{tikzpicture}[node distance=1mm]
\tikzset{every path/.style=thick}

\node (leaf1) [stdnodecompact,label={[tdlabel,  xshift=-0.25em,yshift=0.25em]above:$ $}] {\input{graph0_norm/tables_inc/leaf0_asp}};

\node (leaf0) [stdnodecompact,below=0.15em of leaf1, yshift=0em, xshift=0em, label={[tdlabel,  yshift=-0.0em,xshift=0.0em,]right:$ $}] {\input{graph0_norm/tables_inc/leaf00_asp}};


\node (join) [stdnodecompact,above=0.15em of leaf1,yshift=0.0em, xshift=-0em, label={[tdlabel,  yshift=-0.25em]below:$ $}] {\input{graph0_norm/tables_inc/join_asp}};
\node (r14b) [stdnodecompact,right=0.15em of join,yshift=0em, xshift=-0em, label={[tdlabel,  yshift=-0.25em]below:$ $}] {\input{graph0_norm/tables_inc/r14_asp}};

\node (leaf2) [stdnodecompact,below=0.15em of r14b, xshift=0.0em, yshift=-0.0em, label={[tdlabel,  xshift=+0.1cm,yshift=0.25em]above:$ $}]  {\input{graph0_norm/tables_inc/leaf1_asp}};

%
%
%
\end{tikzpicture}
%
\end{minipage}
\caption{Selected DP tables of~\INCSAT~(left) and~\INC~(right) for
  nice TD~$\TTT'$.}
\label{fig:running1_inc}
\end{figure}%
Since we already explained how to obtain models, we only briefly
describe how we handle the counterwitness part.
Family~$\CCC$ consists of tuples~$(C, \rho)$ where
$C\subseteq \at(\prog) \cap \chi(t)$ is a \emph{counterwitness} in~$t$
to~$M$.
Similar to the rule-state~$\sigma$ the rule-state~$\rho$ for $C$ under
$M$ represents whether rules of the GL reduct~$\prog_t^M$ are either
(i)~satisfied by a superset of $C$ or (ii)~undecided for~$C$.  Thus,
$C$ witnesses the existence of $C' \subsetneq M'$ satisfying
$C' \models (\progtneq{t} \cup \SatRules(\prog_t,\rho))^{M'}$ since
$M$ witnesses a model~$M'\supseteq M$ where
$M' \models \progtneq{t} \cup
\SatRules(\prog_t,\rho)$. 
In consequence, there exists an answer set of $\prog$ if the root
table contains $\langle \emptyset, \emptyfunc, \emptyset \rangle$.  In
order to locally decide rule satisfiability for counterwitnesses, we
require state-programs under witnesses.

\begin{definition}\label{def:bagreduct}%
  Let $\prog$ be a program, $\TTT=(\cdot,\chi)$ be a TD of $I(\prog)$,
  $t$ be a node of $\TTT$, ${\cal P} \subseteq \prog_t$,
  $\rho: \prog_t \rightarrow \NAT_0 \cup \{\infty\}$ be a rule-state
  and $M\subseteq\at(\prog)$.  We define
  \emph{state-program~${\cal P}^{(t,\rho,M)}$} by
  ${[{\cal S}^{(t,\rho)}]}^M$ where
  ${\cal S} \eqdef {\cal P} \cup \{ \leftarrow B_r \mid r \in
  \choice({\cal P}), \rho(r) > 0\}$, and
  $\dot{\cal P}^{(t,\rho,M)}: {\cal P} \rightarrow 2^{{\cal
      P}^{(t,\rho,M)}}$ by
  $\dot{\cal P}^{(t,\rho,M)}(r) \eqdef \{r\}^{(t,\rho,M)}$ for
  $r\in {\cal P}$.
\end{definition}%

\noindent %
We compute a new rule-state~$\rho$ for a counterwitness from an
earlier rule-state, satisfied rules ($\SSR$), and both (a)~``updated''
bounds for weight rules or (b)~``updated'' value representing whether
the head can still be satisfied ($\rho(r)\leq 0$) for choice rules~$r$
($\UpdateRedStates$). Formally,
$\UpdateRedStates(\prog_t, M, C, a) := \sigma'$ depending on an
atom~$a$ with
(a)~$\sigma'(r) \eqdef \wght(r,\{a\} \cap [(B^-_r \setminus M) \cup
(B^+_r \cap C)])$, if $r \in \weight(\prog_t)$; and
(b)~$\Card{\{a\} \cap H_r \cap (M \setminus C)}$, if
$r \in \choice(\prog_t)$.

\begin{algorithm}[t]%
  \KwData{Bag $\chi_t$, bag-rules $\prog_t$, atoms-below $\atto$, child tables $\Tab{}$ of $t$.{~\bf Out:} Tab.~$\tab{t}$.} 
  \tcc{For~$\langle \tuplecolor{\inputPredColor}{M}, \tuplecolor{\statePredColor}{\sigma}, \tuplecolor{\outputPredColor}{\CCC}, c, n
    \rangle$, we only state affected parts (cost~$c$ and
    count~$n$); {\scriptsize``$\dotsc$''} indicates
    computation as before. 
    $\lbag \dotsc \rbag$ denotes a multiset. 
  } 




   %
   %

   \lIf{$\type(t) = \leaf$}{$\tab{t} \eqdef \Big\{ \Big\langle \tuplecolor{\inputPredColor}{\emptyset}, \dotsc, 0, 1 \Big\rangle \Big\}$}
     %
   %
   %
   %
   \uElseIf{$\type(t) = \intr$, $a \in \chi_t \setminus \prog_t$ is introduced and $\tau'\in \Tab{}$}{%

     \makebox[8.01cm][l]{$\tab{t} \eqdef \Big\{ \Big\langle
       \tuplecolor{\inputPredColor}{M}, \dotsc, \oo(\prog, \emptyset,
       \{a\}) + c, n \Big\rangle$}
     $\Bigm|\;\langle \tuplecolor{\inputPredColor}{M}, \tuplecolor{\statePredColor}{\sigma}, \tuplecolor{\outputPredColor}{\CCC},
     c,n \rangle \in \tab{}' \Big\}\;\mcup$

     %
     %
     %
     
     \makebox[8.01cm][l]{$\qquad\;\,\,\Big\{ \Big\langle
       \tuplecolor{\inputPredColor}{\MAI{M}},\dotsc, \oo(\prog, \{a\},
       \{a\}) + c, n \Big\rangle$}
     $\Bigm|\;\langle \tuplecolor{\inputPredColor}{M}, \tuplecolor{\statePredColor}{\sigma}, \tuplecolor{\outputPredColor}{\CCC},
     c,n \rangle \in \tab{}' \Big\}$ \vspace{-0.15em}%
   }%
   %
   %
   \uElseIf{$\type(t) = \intr \text{ or }
     \rem$, removed or introduced $r \in \prog_t$, $\tau'\in
     \Tab{}$}{%
     \makebox[8.01cm][l]{$\tab{t} \eqdef \SB \Big\langle \tuplecolor{\inputPredColor}{M}, \dotsc, c,
       n \Big\rangle$} %
     $\Bigm|\;\langle \tuplecolor{\inputPredColor}{M}, \tuplecolor{\statePredColor}{\sigma}, \tuplecolor{\outputPredColor}{\CCC}, c, n \rangle \in \tab{}', \dotsc
     \Big\}$
     %
   }%
   %
   \uElseIf{$\type(t) = \rem$, $a \notin
     \chi_t$ is removed atom and $\tau'\in \Tab{}$}{%
     \makebox[8.01cm][l]{$\tab{t} \eqdef \cnt(\kmin(\lbag \Big\langle
       \tuplecolor{\inputPredColor}{\MAR{M}}, \dotsc, c, n \Big\rangle$} %
     $\Bigm|\; \langle \tuplecolor{\inputPredColor}{M}, \tuplecolor{\statePredColor}{\sigma}, \tuplecolor{\outputPredColor}{\CCC}, c, n \rangle \in \tab{}'
     \rbag))$ \; \vspace{-0.35em}%
   }%
   %
   %
   \uElseIf{$\type(t) = \join$ and $\tau', \tau'' \in
     \Tab{}$ with $\tau' \neq \tau''$}{%
     \makebox[7.05cm][l]{$\tab{t}\eqdef \cnt(\kmin(\lbag\langle
       \tuplecolor{\inputPredColor}{M}, \dotsc, c' + c'' - \oo(\prog,
       M, \chi_t), n' \cdot n'' \Big\rangle$}\;%
     \makebox[4.5cm][l]{}%
     $\Bigm|\; \langle \tuplecolor{\inputPredColor}{M}, \tuplecolor{\statePredColor}{\sigma'},
     \tuplecolor{\outputPredColor}{\CCC'}, c',n' \rangle \in \tab{}', \langle
     \tuplecolor{\inputPredColor}{M}, \tuplecolor{\statePredColor}{\sigma''}, \tuplecolor{\outputPredColor}{\CCC''}, c'',n''
     \rangle \in \tab{}'' \rbag))$
   }%
   \caption{Algorithm~{\#O\INC}$(t,\chi_t,\prog_t,\atto,\Tab{})$.}
\label{fig:opt}
\end{algorithm}%
%
%
\begin{theorem}\label{thm:inc:correctness}
  The algorithm ${\dpa}_{\INC}$ is correct.
\end{theorem}
\begin{proof}{(Idea)}
  A tuple at a node~$t$ guarantees that there exists a model for the
  ASP sub-program induced by the subtree rooted at~$t$. Since this can
  be done for each node type, we obtain soundness.  Completeness
  follows from the fact that while traversing the tree decomposition
  every answer set is indeed considered. The full proof is rather
  tedious as each node type needs to be investigated separately. For
  more details, we refer the reader to Appendix~\ref{proof:thm:inc:correctness}.
\end{proof}

\begin{theorem}\label{thm:inc:runtime}
  Given a program~$\prog$, algorithm~${\dpa}_{\INC}$ runs in time
  $\bigO{2^{2^{k +2}\cdot\ell^{k+1}}\cdot \CCard{I(\prog)}}$, where
  $k\eqdef\tw{I(\prog)}$, and
  $\ell\eqdef max\SB 3, \bnd(r) \SM r \in \weight(\prog)\SE$.
\end{theorem}
\shortversion{%
  \vspace{-0.25em}
} %
\begin{proof}
  We refer the reader to Appendix~\ref{proof:inc:runtime}.
\end{proof}

\smallskip\noindent The runtime bounds stated in
Theorem~\ref{thm:inc:runtime} appear to be worse than in
Theorem~\ref{thm:prim:runtime}. However,
$\tw{I(\prog)} \leq \tw{P(\prog)} + 1$ and
$\tw{P(\prog)} \geq \max \SB \Card{\at(r)} \SM r \in \prog \SE$ for a
given program~$\prog$. Further, there are programs where
$\tw{I(\prog)}=1$, but $\tw{P(\prog)}=k$,~e.g., a program consisting
of a single rule~$r$ with $\Card{\at(r)} = k$.
Consequently, worst-case runtime bounds of~${\dpa}_{\PRIM}$ are at
least double-\allowbreak{}exponential in the rule size and
${\dpa}_{\PRIM}$ will perform worse than ${\dpa}_{\INC}$ on input
programs containing large rules. However, due to the rule-states, data
structures of ${\dpa}_{\INC}$ are much more complex than of
${\dpa}_{\PRIM}$. In consequence, we expect ${\dpa}_{\PRIM}$ to
perform better in practice if rules are small and incidence and
primal treewidth are therefore almost equal.
In summary, we have a trade-off between (i)~a more general parameter
decreasing the theoretical worst-case runtime and (ii)~less
complex data structures decreasing the practical overhead to solve 
\AspComp.
\newcommand{\AAA}{\ensuremath{\mathcal{A}}}%
\subsection{Extensions for Optimization and Counting}\label{sec:alg:ext}
In order to find an answer set of a program with optimization
statements or the number of optimal answer sets (\AspCountO), we
extend our algorithms \PRIM and \INC. Therefore, we augment tuples
stored in tables with an integers~$c$ and $n$ describing the cost and
the number of witnessed sets. Due to space restrictions, we only
present adaptions for~\INC.
\begin{figure*}[t]
\centering
\includegraphics[scale=0.242]{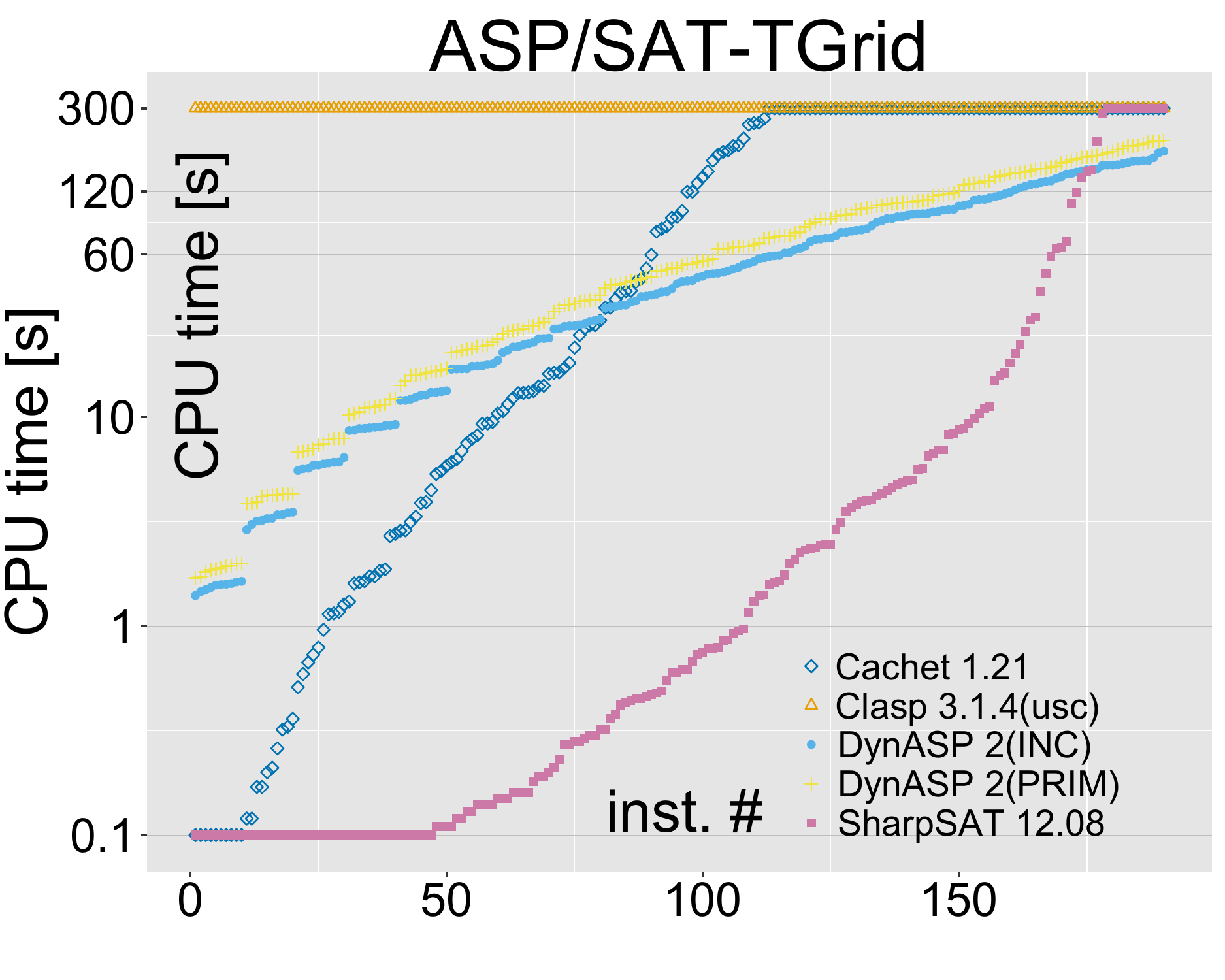}
\includegraphics[scale=0.242]{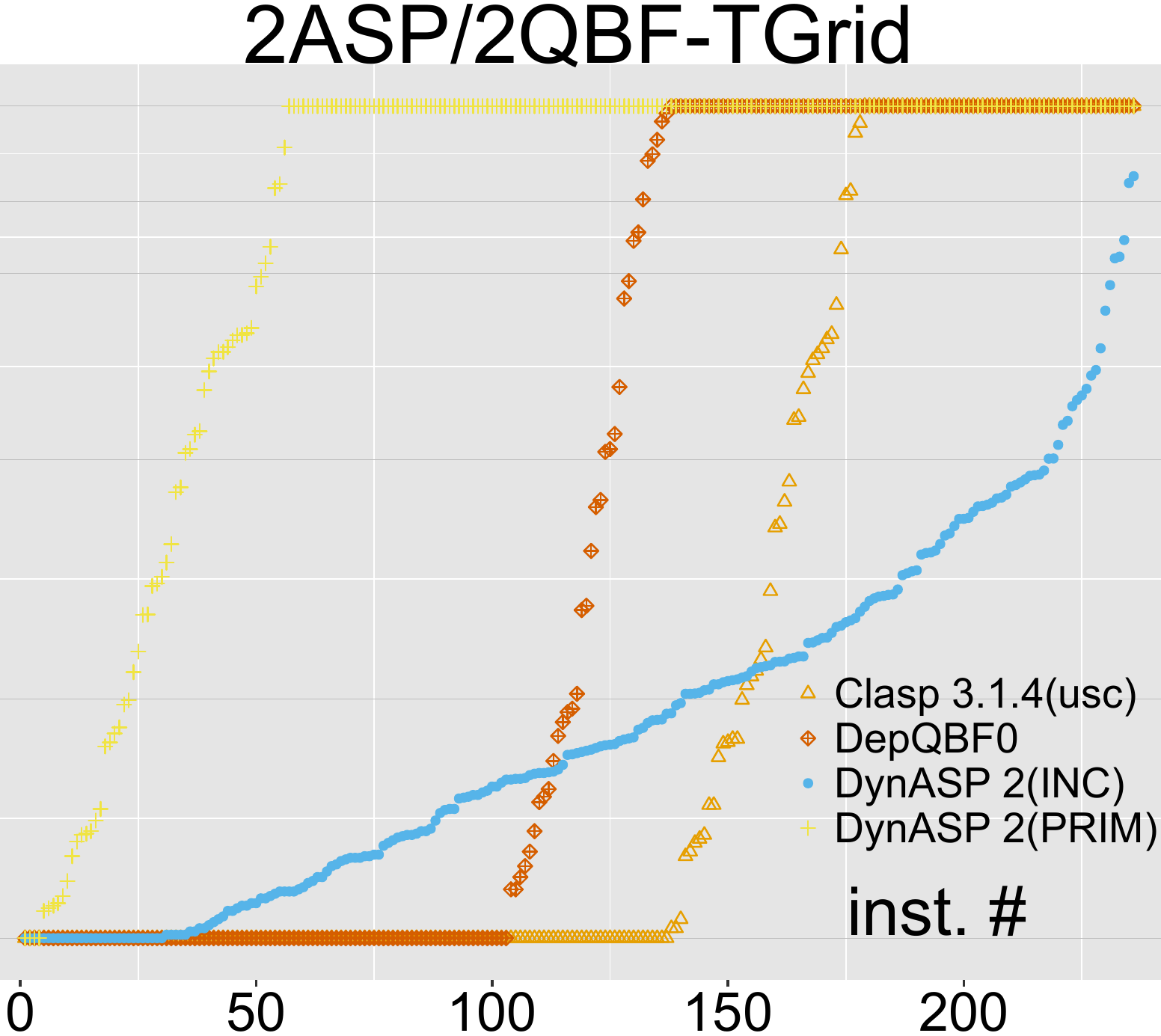}
\includegraphics[scale=0.242]{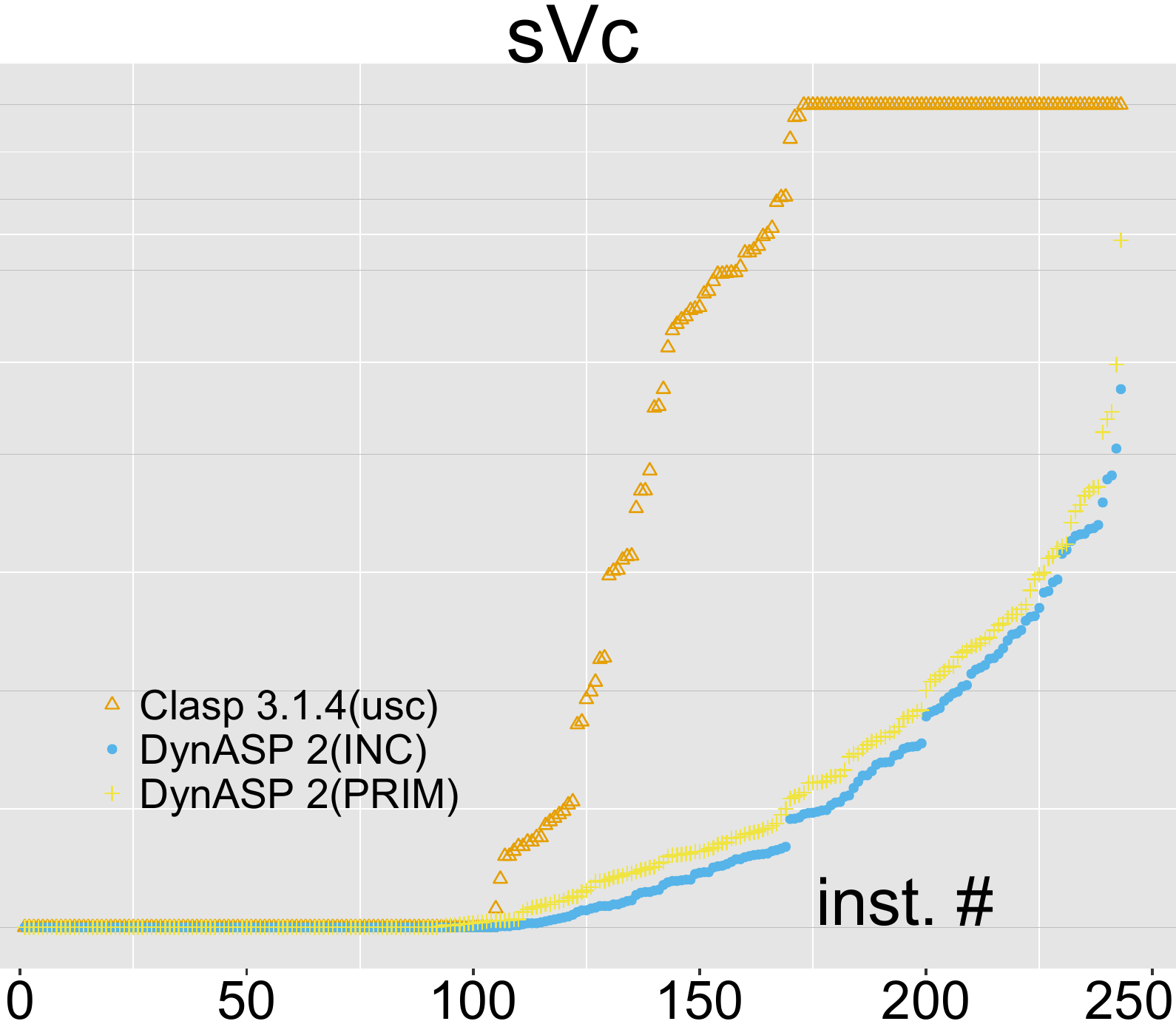}
\caption{Results of randomly generated and selected real-world instances.}
  \label{fig:random}
\end{figure*}%
We state which parts of \INC we adapt to compute the number of optimal
answer sets in Algorithm~\ref{fig:opt} ({\#O\INC}).
To slightly simplify the presentation of optimization rules, we assume
without loss of generality that whenever an atom~$a$ is introduced in
bag~$\chi(t)$ for some node~$t$ of the TD, the optimization rule~$r$,
where $a$ occurs, belongs to the bag~$\chi(t)$.
First, we explain how to handle costs making use of function
$\oo(\prog, M, A)$ as defined in Section~\ref{sec:preliminaries}.
In a leaf (Line~1) we set the (current) cost to~$0$.  If we introduce
an atom~$a$ (Line~2--4) the cost depends on whether $a$ is set to true
or false in $M$ and we add the cost of the ``child'' tuple. Removal of
rules (Line~5--6) is trivial, as we only store the same values. If we
remove an atom (Line~7--8), we compute the minimum costs only for
tuples~$\langle \MAR{M}, \sigma, \CCC, c, n \rangle$ where $c$ is
minimal among $\MAR{M}$, $\sigma$, $\CCC$, that is, for a
multiset~$\SSS$ we let
$\kmin(\SSS) \eqdef \lbag \langle \MAR{M}, \sigma, \CCC, c, n \rangle
\SM c = \min \SB c' : \langle \MAR{M}, \sigma, \CCC, c', \cdot \rangle
\in \SSS \SE, \langle \MAR{M}, \sigma, \CCC, c, n \rangle \in \SSS
\rbag$. We require a multiset notation for counting (see below). If we
join two nodes (Line~9--11), we compute the minimum value in the table
of one child plus the minimum value of the table of the other child
minus the value of the cost for the current bag, which is exactly the
value we added twice.
Next, we explain how to handle the number of witnessed sets that are
minimal with respect to the cost. In a leaf (Line~1), we set the
counter to~$1$. If we introduce/remove a rule or introduce an atom
(Line~2--6), we can simply take the number~$n$ from the child. If we
remove an atom (Line~7--8) we first obtain a multiset from
computing~$\kmin$, which can contain several tuples for
$\MAR{M},\sigma,\CCC,c$ as we obtained $\MAR{M}$ either from
$M \setminus \{a\}$ if $a \in M$ or $M$ if $a \notin M$ giving rise
multiple solutions, that is,
$\cnt(\SSS) \eqdef \SB \langle M,\sigma,\CCC,c,\sum_{\langle
  M,\sigma,\CCC,c,n'\rangle \in \SSS} n' \rangle \SM \langle M,\sigma,
\CCC,c, n\rangle \in \SSS \SE$. If we join nodes (Line~7--9), we
multiply the number~$n'$ from the tuple of one child with the
number~$n''$ from the tuple of the other child, restrict results with
respect to minimum costs, and sum up the resulting numbers.

\begin{corollary}\label{cor:oinc}
  Given a program~$\prog$, algorithm~{\normalfont{${\#O\INC}$}} runs
  in time
  $\bigO{log(m) \cdot 2^{2^{k+2}\cdot \ell^{k+1}}
    \CCard{I(\prog)}^2}$, where $k\eqdef \tw{I(\prog)}$,
  $\ell \eqdef max\{3, \bnd(r) : r \in \weight(\prog)\}$, and
  $m \eqdef \Sigma_{r \in \opt(\prog)}\wght(r)$.
\end{corollary}

\section{Experimental Evaluation}
\label{sec:evaluation}

We implemented the algorithms~${{\dpa}_{\PRIM}}$ and~${{\dpa}_{\INC}}$
into a prototypical solver {\dynasp{\cdot}} and performed experiments
to evaluate its runtime behavior.  Clearly, we \emph{cannot} hope to
solve programs with graph representations of high treewidth. However,
programs involving real-world graphs such as graph problems on transit
graphs admit TDs of small width.
We used both random and structured instances for our benchmarks. We
refer to Appendix~\ref{appendix:experiments} for instance, machine and
solver configurations and descriptions.
The random instances (\pname{Sat-TGrid}, \pname{2QBF-TGrid},
\pname{ASP-TGrid}, \pname{2ASP-TGrid}) were designed to have a high
number of variables and solutions and treewidth at most three.
The structured instances model various graph problems (\pname{2Col},
\pname{3Col}, \pname{Ds}, \pname{St} \pname{cVc}, \pname{sVc}) on real
world mass transit graphs.
For a graph, program \pname{2Col} counts all 2-colorings, \pname{3Col}
counts all 3-colorings, \pname{Ds} counts all minimal dominating sets,
\pname{St} counts all Steiner trees, \pname{cVc} counts all
cardinality-minimal vertex covers, and \pname{sVc} counts all
subset-minimal vertex covers.
\newcommand{\blahtab}[1]{{\tiny{#1}}}
\begin{table}[t]
  \centering
  \small
  \begin{tabular}{@{}l@{\hspace{0.1em}}|@{\hspace{0.2em}}
	r@{\hspace{0.2em}}r@{\hspace{0.3em}}
	r@{}			  r@{\hspace{0.3em}}
	r@{\hspace{0.2em}}r@{\hspace{0.3em}}
	r@{\hspace{0.2em}}r@{\hspace{0.3em}}
	r@{\hspace{0.2em}}r@{\hspace{0.3em}}
	r@{\hspace{0.2em}}r@{}}
    \toprule 
            & \multicolumn{2}{c}{\pname{2Col}} & 
					\multicolumn{2}{c}{\pname{3Col}} & 
					\multicolumn{2}{c}{\pname{Ds}} & 
					\multicolumn{2}{c}{\pname{St}} & 
					\multicolumn{2}{c}{\pname{cVc}} & 
					\multicolumn{2}{c}{\pname{sVc}}\\ 
    \midrule
    {\blahtab{Clasp(usc)}} &  31.72 & (21) & 0.10 &(0)  & 8.99 & (3) & 0.21 & (0)  & 29.88 & (21) & 98.34 & (71) \\
    \blahtab{DynASP2(PRIM)} & 1.54 &(0) & 0.53 &(0) & 0.68 & (0) & 79.36 & (221) & 0.99 & (0) & 1.30 & (0)\\
    \blahtab{DynASP2(INC)} & 1.43 &(0) & 0.58 &(0) & 0.54 & (0) & 115.02 & (498) & 0.68 & (0) & 0.78 & (0) \\
    \bottomrule
  \end{tabular}
  \caption{Runtimes (given in sec.; \#timeouts in brackets) on real-world instances.}
  \label{tab:real_world}
  \vspace{-1.0em}
\end{table}%
In order to draw conclusions about the efficiency of DynASP2, we
mainly inspected the cpu running time and number of timeouts using the
average over three runs per instance (three fixed seeds allow certain
variance~\cite{AbseherEtAl15a} for heuristic TD computation).  We
limited available memory (RAM) to 4GB (to run \sharpsat on large
instances), and cpu time to 300 seconds, and then compared DynASP2
with the dedicated \#SAT solvers \sharpsat~\cite{sat:Thurley06} and
Cachet~\cite{sat:SangBBKP04}, the QBF solver \depqbfz, and the ASP
solver
Clasp\shortversion{.}\longversion{~\cite{GebserKaufmannSchaub12a}.}
Figure~\ref{fig:random} illustrates runtime results as a cactus plot.
Table~\ref{tab:real_world} reports on the average running times,
numbers of solved instances and timeouts on the structured instance
sets.
\vspace{-0.6em}
\paragraph{Summary.}
Our empirical benchmark results confirm that DynASP2 exhibits
competitive runtime behavior if the input instance has small
treewidth. Compared to state-of-the-art \ASP and \QBF solvers, DynASP2
has an advantage in case of many solutions, whereas Clasp and \depqbfz
perform well if the number of solutions is relatively small. However,
DynASP2 is still reasonably fast on structured instances with few
solutions as it yields the result mostly within less than 10 seconds. We
observed that {\INC} seems to be the better algorithm in our setting, indicating
that the smaller width obtained by decomposing the incidence graph generally
outweighs the benefits of simpler solving algorithms for the primal graph.
However, if {\INC} and {\PRIM} run with graphs of similar width, {\PRIM}
benefits from its simplicity. A comparison to existing \#SAT solvers suggests
that, on random instances, they have a lower overhead (which is not surprising,
since our algorithms are built for ASP), but, after about 150 seconds, our
algorithms were still able to solve more instances than all other \#SAT
competitors.

\section{Conclusion}
\label{sec:conclusions}
In this paper, we presented novel DP algorithms for ASP, extending
previous work~\cite{JaklPichlerWoltran09} in order to cover the full
ASP syntax. Our algorithms are based on two graph representations of
programs and run in linear time with respect to the treewidth of these
graphs and weights used in the program.  Experiments indicate that our
approach seems to be suitable for practical use, at least for certain
classes of instances with low treewidth, and hence could fit into a
portfolio-based solver.

%
\bibliographystyle{plain}
\longversion{%

%
}
\shortversion{%
  
}

\newpage
\appendix

\section{Additional Examples}

In the following example, we briefly describe how we compute
counterwitnesses using Algorithm~\ref{fig:incinc} (\algo{INC}) for
selected interesting cases. The example is similar to
Example~\ref{ex:prim:min}, which, however, describes handling
counterwitnesses for Algorithm \algo{PRIM}.

\begin{example}
  We consider $\prog$ of Example~\ref{ex:running1} and
  ${\cal T}'=(\cdot,\chi)$ of Figure~\ref{fig:running1_inc} and
  explain how we compute tables~$\tab{1}$, $\ldots$, $\tab{18}$ in
  Figure~\ref{fig:running1_inc}~(right) 
  using $\dpa_\INC$.
  Table~$\tab{1}=\langle \emptyset, \emptyfunc, \emptyset \rangle$ as
  $\type(t_1)=\leaf$. Node~$t_2$ introduces atom~$c$, resulting in
  table
  $\{\langle \{c\}, \emptyfunc, \{(\emptyset, \emptyfunc)\} \rangle,
  \langle \emptyset, \emptyfunc, \emptyset \rangle \}$.  Then,
  node~$t_3$ introduces rule~$r_1$ and node~$t_4$ introduces
  rule~$r_2$. As a result, table~$\tab{4}$ additionally contains
  computed rule-states (see~$\SSR$) for witnesses and counterwitnesses
  of~$\tab{3}$.  Node~$t_5$ introduces atom~$b$, while~$t_6$
  removes~$b$.
  Next, we focus on table~$\tab{6}$, since rule-states for
  counterwitnesses require updates for choice rule~$r_1$
  (see~$\UpdateRedStates$). Witness $M_{6.2}=\{c\}$ is obtained by
  extending some witness~$M_{5.i}\supseteq\{b\}$ of~$\tab{5}$.  For
  counter\-witness~$C_{6.2.1}=\{c\}$ we require to
  remember~$\sigma_{6.2.1}(r_1)=1$ (see~$\UpdateRedStates$),
  since~$t_6$ removes~$b$ and~$C_{6.2.1}$ stems from
  some~$C_{5.i.{j_1}}$ with $b\not\in C_{5.i.{j_1}}$. The
  set~$C_{5.i.{j_1}}$ \emph{cannot} be a model of the GL
  reduct~${\{r_1\}}^{M_{5.i}}$ unless~$r_1$ is satisfied because of its body, 
  since $b\in M_{5.i}$ and $b\not\in C_{5.i.{j_1}}$. For
  choice rule~$r_1$, $\sigma_{6.2.1}(r_1) \neq \infty$ and
  $\sigma_{6.2.1}(r_1) \neq 0$ indicates that we can satisfy~$r_1$ 
  only by $B^+(r_1) \setminus M \neq \emptyset$ (see
  ${\cal P} \cup \SB \cdots, \rho(r) > 0 \SE$ in
  Definition~\ref{def:bagreduct}). The remaining
  counterwitness~$C_{6.2.2}=\emptyset$ was obtained by
  some~$C_{5.i.{j_2}}$ with $b\not\in C_{5.i.{j_2}}$, since
  $\sigma_{6.2.2}(r_2) = 0$). Further, $C_{6.2.3}=\emptyset$ stems
  from $C_{5.i.{j_3}}\supseteq\{b\}$, since $\sigma_{6.2.3}(r_2) = 1$.
\end{example}

\section{Omitted Proofs}

\subsection{Proof of Theorem~\ref{thm:prim:runtime} (Correctness
  result of $\algo{PRIM}$)}\label{proof:thm:prim:runtime}

\begin{proposition}\label{prop:prim:correctness}
  The algorithm ${\dpa}_{\PRIM}$ is correct.
\end{proposition}
\shortversion{\begin{proof}[Sketch]}
\longversion{\begin{proof}[Proof (Sketch).]}
  Let $\prog$ be the given program and $\TTT=(T,\chi)$ the TD, where
  $T = (N,\cdot, n)$.  We obtain correctness by slightly modifying the
  proof of Theorem~\ref{thm:prim:runtime} as well as relevant
  definitions and propositions following
  Appendix~\ref{proof:thm:inc:correctness}. More precisely, we drop
  the mappings~$\sigma$ and relevant conditions for mappings~$\sigma$
  and replace them by satisfiability of the respective rules. By
  definition of a primal graph of a program, we know that for every
  rule~$r \in \prog$ there is a node~$t \in N$ such that
  $\chi(t) \subseteq \at(r)$. Hence, for a node~$t$ we can decide
  satisfiability of a rule directly, if bag~$\chi(t)$ contains all
  atoms of a rule, when computing the tables. We directly obtain
  completeness and soundness, which yields the proposition.
\end{proof}

\begin{proposition}\label{prop:prim:space}
  Given a program~$\prog$ and a TD ${\cal T} = (T, \chi)$ of the
  primal graph~$P(\prog)$ of width~$k$ with $T=(N, \cdot, \cdot)$. For
  every node $t\in N$, there are at most~$2^{k+1}\cdot2^{2^{k+1}}$
  tuples in table~$\tab{t}$, which is constructed by
  algorithm~${\dpa}_{\INC}$.
\end{proposition}
\begin{proof}
  Let $\prog$ be a program, $P(\prog)$ its primal graph, and
  $\TTT=(T,\chi)$ a TD of~$P(\prog)$ with $T = (N,\cdot,\cdot)$. For
  every node~$t \in T$, we have by definition of a tree decomposition
  and its width a maximum bag size of~$k+1$,~i.e.,
  $\Card{\chi(t)} - 1 \leq k$.  Therefore, we can have $2^{k+1}$ many
  witnesses and for each witness a subset of the set of witnesses
  consisting of at most~$2^{2^{k+1}}$ many counterwitnesses.
  Consequently, there are at most~$2^{k+1} \cdot 2^{2^{k+1}}$ tuples
  per node. Hence, the proposition is true.
\end{proof}

\noindent
Now, we are in situation to prove Theorem~\ref{thm:prim:runtime}.

\shortversion{\begin{proof}[of Theorem~\ref{thm:prim:runtime}]}
\longversion{\begin{proof}[Proof of Theorem~\ref{thm:prim:runtime}.]}
  Let $\prog$ be a program, $I(\prog)=(V,\cdot)$ its incidence graph,
  and $k$ be the treewidth of $P(\prog)$.
  Proposition~\ref{prop:prim:correctness} establishes correctness.
  Then, we can compute in time~$2^{\bigO{k^3}} \cdot \Card{V}$ a TD of
  width at most~$k$~\citesec{Bodlaender96}.  We take such a TD and
  compute in linear time a nice TD~\citesec{Kloks94a}. Let
  $\TTT = (T,\chi)$ be such a nice TD with $T = (N,\cdot,n)$.
  Since the number of nodes in~$N$ is linear in the graph size and
  since for every node~$t \in N$ the table~$\tab{t}$ is bounded by
  $2^{k+1}\cdot2^{2^{k+1}}$ according to
  Proposition~\ref{prop:prim:space}, we obtain a running time
  of~$\bigO{2^{2^{k+2}}\cdot\CCard{P(\prog)}}$. Consequently, the
  theorem sustains.
\end{proof}

\subsection{Proof of Theorem~\ref{thm:inc:correctness} (Correctness result of \algo{INC})}\label{proof:thm:inc:correctness}
In the following, we provide insights on the correctness of
Algorithm~\ref{fig:incinc} (\algo{INC}).  The correctness proof of
these algorithms need to investigate each node type separately.  We
have to show that a tuple at a node~$t$ guarantees existence of a
model for the program~$\prog_{\leq t}$, proving soundness. Conversely,
one can show that each candidate answer set is indeed evaluated while
traversing the TD, which provides completeness.
We employ this idea using the notions of (i)~\emph{partial solutions}
consisting of \emph{partial models} and the notion of (ii)~\emph{local
  partial solutions}.

\begin{definition}\label{def:partial-model}
  Let $\prog$ be a program, $\calT = (T, \chi)$ be a TD of the
  incidence graph~$I(\prog)$ of~$\prog$, where $T=(N,\cdot,\cdot)$,
  and $t\in N$ be a node. Further, let ${M},C \subseteq \atto$ be sets
  and ${\sigma}: \progt{t} \rightarrow \NAT_0 \cup \{\infty \}$ a
  mapping. The tuple $({C}, {\sigma})$ is a \emph{partial model
    for~$t$ under $M$} if the following conditions hold:
  \begin{enumerate}
  \item ${C} \models (\progtneq{t})^M$,
  \item for $r \in \progt{t}$ we have ${\sigma}(r) = 0$ or
    ${\sigma}(r) = \infty$,
  \item
    \begin{enumerate}
    \item for $r \in \disj(\progt{t})$ we have
      $B^-_r \cap M \neq \emptyset$ or
      $B^+_r \cap \atto \not\subseteq C$ or
      $H_r \cap C \neq \emptyset$ if and only if $\sigma(r)=\infty$,
    \item for $r \in \weight(\progt{t})$ we have 
      $\wght(r, (\at(r) \setminus \atto) \cup (B^-_r \setminus M) \cup
      (B^+_r \cap C)) < \bnd(r)$ or $H_r \cap C \neq \emptyset$ if and
      only if $\sigma(r)=\infty$, and
    \item for $r \in \choice(\progt{t})$ we have
      $B^-_r \cap M \neq \emptyset$ or
      $B^+_r \cap \atto \not\subseteq C$ or both $H_r \subseteq \atto$
      and $H_r \cap (M \setminus C) = \emptyset$ if and only if
      ${\sigma(r)} = \infty$.
    \end{enumerate}
  \end{enumerate}

\end{definition}

\begin{definition}\label{def:partialsol}
  Let $\prog$ be a program, $\calT = (T, \chi)$ where
  $T=(N,\cdot,n)$ be a TD of $I(\prog)$, and $t\in
  N$ be a node.  A \emph{partial solution for~$t$} is a tuple $({M},
  {\sigma}, {{\cal C}})$ where $({M},
  {\sigma})$ is a partial model under~$M$ and ${{\cal
      C}}$ is a set of partial models $({C}, {\rho})$ under
  ${M}$ with ${C} \subsetneq {M}$.
\end{definition}

\noindent The following lemma establishes correspondence between
answer sets and partial solutions.

\begin{lemma}\label{prop:partialsol_corr}
  Let $\prog$ be a program, $\calT = (T, \chi)$ be a TD of the
  incidence graph~$I(\prog)$ of program~$\prog$, where
  $T=(\cdot,\cdot,n)$, and $\chi(n) = \emptyset$.  Then, there exists
  an answer set ${M}$ for $\prog$ if and only if there exists a
  partial solution $\tabval = ({M}, \sigma, \emptyset)$ with
  $\sigma^{-1}(\infty) = \prog$ for root~$n$.
\end{lemma}
\begin{proof}
  Given an answer set ${M}$ of $\prog$ we construct
  $\tabval = ({M}, {\sigma}, \emptyset)$ with
  ${\sigma}(r) \eqdef \infty$ for $r\in \prog$ such that $\tabval$ is
  a partial solution for $n$ (according to
  Definition~\ref{def:partialsol}).  For the other direction,
  Definitions~\ref{def:partial-model} and \ref{def:partialsol}
  guarantee that ${M}$ is an answer set if there exists some
  tuple~$\tabval$. In consequence, the lemma holds.
\end{proof}

\noindent Next, we require the notion of local partial solutions
corresponding to the tuples obtained in Algorithm~\ref{fig:incinc}.

\begin{definition}
  Let $\prog$ be a program, $\calT = (T, \chi)$ a TD of $I(\prog)$,
  where $T=(N,\cdot,n)$, $t\in N$ be a node,
  ${M},C \subseteq \at(\prog)$ sets, and
  ${\sigma}: \prog \rightarrow \NAT_0 \cup \{\infty\}$ be a mapping.
  We define the \emph{local rule-state}
  $\sigma^{t,M,C} \eqdef \MARR{({\sigma} \squplus
    \sigma')}{\progtneq{t}}$ for ${C}$ under $M$ of node $t$ where
  $\sigma' : \prog_t \rightarrow \NAT_0 \cup \{ \infty \}$ by

  \begin{align*}
    \sigma'(r)\eqdef \begin{cases}
      \wght(r, (\atto \setminus \chi(t)) \cap [(B^-_r \setminus {M}) \cup (B^+_r \cap {C})])& r \in \weight(\prog_t)\quad\\
      \card{(\atto \setminus \chi(t)) \cap H_r \cap ({M} \setminus
        C)}& r \in \choice(\prog_t)
  \end{cases}
  \end{align*}
\end{definition}%
\begin{definition}\label{def:localpartialsol}
  Let $\prog$ be a program, $\calT = (T, \chi)$ a TD of the incidence
  graph~$I(\prog)$, where $T=(N,\cdot,n)$, and $t\in N$ be a node. A
  tuple $\tabval = \langle M, \sigma, {\cal C} \rangle$ is a
  \emph{local partial solution} for $t$ if there exists a partial
  solution
  ${\hat \tabval} = ({\hat M}, {\hat \sigma}, {\hat {\cal C}})$
  for~$t$ such that the following conditions hold:
\begin{enumerate}
\item $M = {\hat M} \cap \chi(t)$,
\item $\sigma = {\hat \sigma}^{t,{\hat M}, {\hat M}}$, and
\item
  ${\cal C} = \SB \langle {\hat C} \cap \chi(t), {\hat \rho}^{t,{\hat
      M},{\hat C}} \rangle \SM ({\hat C}, {\hat \rho}) \in {\hat {\cal
      C}} \SE$.
\end{enumerate}
We denote by ${\hat \tabval}^t$ the local partial solution $\tabval$
for~$t$ given partial solution ${\hat \tabval}$.

\end{definition}

\noindent The following proposition provides justification that it
suffices to store local partial solutions instead of partial solutions
for a node~$t \in N$.

\begin{lemma}\label{prop:partiallocalsol_corr}
  Let $\prog$ be a program, $\calT = (T, \chi)$ a TD of $I(\prog)$,
  where $T=(N,\cdot,n)$, and $\chi(n) = \emptyset$.  Then, there
  exists an answer set for $\prog$ if and only if there exists a local
  partial solution of the
  form~$\langle \emptyset, \emptyset, \emptyset \rangle$ for the
  root~$n \in N$.
\end{lemma}
\begin{proof}
  Since $\chi(n) = \emptyset$, every partial solution for the root~$n$
  is an extension of the local partial solution~$\tabval$ for the
  root~$n \in N$ according to Definition~\ref{def:localpartialsol}. By
  Lemma~\ref{prop:partialsol_corr}, we obtain that the lemma is true.
\end{proof}

\noindent
In the following, we abbreviate atoms occurring in bag~$\chi(t)$
by~$\at_t$,~i.e., $\at_t \eqdef \chi(t) \setminus \prog_t$.

\begin{proposition}[Soundness]\label{thm:soundness}
  Let $\prog$ be a program, $\calT = (T, \chi)$ a TD of incidence
  graph~$I(\prog)$, where $T=(N,\cdot,\cdot)$, and $t\in N$ a node.
  Given a local partial solution $\tabval'$ of child table $\tab{}'$
  (or local partial solution $\tabval'$ of table $\tab{}'$ and local
  partial solution $\tabval''$ of table $\tab{}''$), each tuple
  $\tabval$ of table $\tab{t}$ constructed using table algorithm
  $\INC$ is also a local partial solution.
\end{proposition}
\begin{proof}
  Let $\tabval'$ be a local partial solution for $t'\in N$ and
  $\tabval$ a tuple for node~$t \in N$ such that $\tabval$ was derived
  from~$\tabval'$ using table algorithm \INC. Hence, node~$t'$ is the
  only child of $t$ and~$t$ is either removal or introduce node.

  Assume that $t$ is a removal node and
  $r\in \prog_{t'}\setminus \prog_t$ for some rule~$r$.
  Observe that $\tabval = \langle M, \sigma, {\cal C} \rangle$ and
  $\tabval' = \langle M, \sigma', {\cal C'} \rangle$ are the same in
  witness~$M$.  According to Algorithm~\ref{fig:incinc} and since
  $\tabval$ is derived from $\tabval'$, we have $\sigma'(r) =
  \infty$. Similarly, for any
  $\langle C', \rho' \rangle \in {\cal C'}$, $\rho'(r) = \infty$.
  Since $\tabval'$ is a local partial solution, there exists a partial
  solution~${\hat \tabval'}$ of $t'$, satisfying the conditions of
  Definition~\ref{def:localpartialsol}.  Then, ${\hat \tabval'}$ is
  also a partial solution for node $t$, since it satisfies all
  conditions of Definitions~\ref{def:partial-model} and
  \ref{def:partialsol}.  Finally, note that
  $\tabval = ({\hat {\tabval}'})^t$ since the projection of
  ${\hat \tabval'}$ to the bag $\chi(t)$ is $\tabval$ itself. In
  consequence, the tuple~$\tabval$ is a local partial solution.

  For~$a\in\at_{t'}\setminus\at_t$ as well as for introduce nodes, we
  can analogously check the proposition.

  Next, assume that $t$ is a join node. Therefore, let $\tabval'$ and
  $\tabval''$ be local partial solutions for $t',t''\in N$,
  respectively, and $\tabval$ be a tuple for node $t\in N$ such that
  $\tabval$ can be derived using both $\tabval'$ and $\tabval''$ in
  accordance with the \INC algorithm.  Since $\tabval'$ and
  $\tabval''$ are local partial solutions, there exists partial
  solution
  ${\hat \tabval'} = ({\hat M'}, {\hat \sigma'}, {\hat {\cal C'}})$
  for node $t'$ and partial solution
  ${\hat \tabval''} = ({\hat M''}, {\hat \sigma''}, {\hat {\cal
      C''}})$ for node $t''$.  Using these two partial solutions, we
  can construct
  ${\hat \tabval} = ({\hat M'} \cup {\hat M''}, {\hat \sigma'}
  \squplus {\hat \sigma''}, {\hat {\cal C'}} \bowtie {\hat {\cal
      C''}})$ where $\bowtie (\cdot, \cdot)$ is defined in accordance
  with Algorithm~\ref{fig:incinc} as follows:
  \begin{align*} {\hat {\cal C'}} \bowtie {\hat {\cal C''}} \eqdef
    &\SB({\hat C'} \cup {\hat C''}, {\hat \rho'} \squplus {\hat
      \rho''}) \SM ({\hat C'}, {\hat \rho'}) \in {\hat {\cal C'}},
    ({\hat C''}, {\hat \rho''}) \in {\hat {\cal C''}}, {\hat C'}
      \cap \at_t = {\hat C''} \cap \at_t\SE\cup \\
    &\SB({\hat C'} \cup {\hat M''}, {\hat \rho'} \squplus 
      {\hat \sigma''}) \SM ({\hat C'}, {\hat \rho'}) \in {\hat {\cal C'}}, 
      {\hat C'} \cap \at_t = {\hat M''} \cap \at_t\SE \cup \\
    &\SB({\hat M'} \cup {\hat C''}, {\hat \sigma'} \squplus 
      {\hat \rho''}) \SM ({\hat C''}, {\hat \rho''}) \in {\hat {\cal C''}}, 
      {\hat M'} \cap \at_t = {\hat C''} \cap \at_t \SE.
  \end{align*}
  Then, we check all conditions of Definitions~\ref{def:partial-model}
  and \ref{def:partialsol} in order to verify that ${\hat
    \tabval}$ is a partial solution for
  $t$. Moreover, the projection ${\hat \tabval}^t$ of ${\hat
    \tabval}$ to the bag $\chi(t)$ is exactly
  $\tabval$ by construction and hence, $\tabval = {\hat
    \tabval}^t$ is a local partial solution.

  Since we have provided arguments for each node type, we established
  soundness in terms of the statement of the proposition.

\end{proof}

\begin{proposition}[Completeness]\label{prop:completeness}
  Let $\prog$ be a program, $\calT = (T, \chi)$ where
  $T=(N,\cdot,\cdot)$ be a TD of $I(\prog)$ and $t\in N$ be a
  node. Given a local partial solution $\tabval$ of table $\tab{t}$,
  either $t$ is a leaf node, or there exists a local partial solution
  $\tabval'$ of child table $\tab{}'$ (or local partial solution
  $\tabval'$ of table $\tab{}'$ and local partial solution~$\tabval''$
  of table $\tab{}''$) such that $\tabval$ can be constructed by
  $\tabval'$ (or $\tabval'$ and $\tabval''$, respectively) and using
  table algorithm~${\INC}$.
\end{proposition}
\begin{proof}
  Let $t\in N$ be a removal node and
  $r\in \prog_{t'} \setminus \prog_t$ with child node~$t'\in N$.  We
  show that there exists a tuple~$\tabval'$ in table~$\tab{t'}$ for
  node $t'$ such that $\tabval$ can be constructed using $\tabval'$ by
  \INC (Algorithm~\ref{fig:incinc}). Since $\tabval$ is a local
  partial solution, there exists a partial solution
  ${\hat \tabval} = ({\hat M}, {\hat \sigma}, {\hat {\cal C}})$ for
  node~$t$, satisfying the conditions of
  Definition~\ref{def:localpartialsol}.  Since $r$ is the removed
  rule, we have ${\hat \sigma}(r) = \infty$. By similar arguments, we
  have $\hat \rho(r) = \infty$ for any tuple
  $({\hat C}, {\hat \rho}) \in {\hat {\cal C}}$. Hence,
  ${\hat \tabval}$ is also a partial solution for~$t'$ and we define
  $\tabval' \eqdef {\hat \tabval}^{t'}$, which is the projection of
  ${\hat \tabval}$ onto the bag of~$t'$. Apparently, the
  tuple~$\tabval'$ is a local partial solution for node $t'$ according
  to Definition~\ref{def:localpartialsol}. Then, $\tabval$ can be
  derived using \INC algorithm and $\tabval'$.  By similar arguments,
  we establish the proposition for~$a\in \at_{t'}\setminus\at_t$ and
  the remaining (three) node types. Hence, the propositions sustains.
\end{proof}

\noindent
Now, we are in situation to prove Theorem~\ref{thm:inc:correctness}.

\shortversion{\begin{proof}[of Theorem~\ref{thm:inc:correctness}]}
\longversion{\begin{proof}[Proof of Theorem~\ref{thm:inc:correctness}.]}
  %
  We first show soundness. Let $\TTT = (T, \chi)$ be the given TD,
  where $T = (N,\cdot ,n)$.
  By Lemma~\ref{prop:partiallocalsol_corr} we know that there is an
  answer set for $\prog$ if and only if there exists a local partial
  solution for the root~$n$. Note that the tuple is of the
  form~$\langle \emptyset, \emptyfunc, \emptyset \rangle$ by
  construction.  Hence, we proceed by induction starting from the leaf
  nodes. In fact, the
  tuple~$\langle \emptyset, \emptyset, \emptyset \rangle$ is trivially
  a partial solution by Definitions~\ref{def:partial-model}
  and~\ref{def:partialsol} and also a local partial solution of
  $\langle \emptyset, \emptyset, \emptyset \rangle$ by
  Definition~\ref{def:localpartialsol}.  We already established the
  induction step in Proposition~\ref{thm:soundness}.
  Hence, when we reach the root~$n$, when traversing the TD in
  post-order by Algorithm~$\dpa_{\INC}$, we obtain only valid tuples
  inbetween and a tuple of the
  form~$\langle \emptyset, \emptyset, \emptyset \rangle$ in the table
  of the root~$n$ witnesses an answer set.
  Next, we establish completeness by induction starting from the
  root~$n$. Let therefore, $M$ be an arbitrary answer set of~$\prog$.
  By Lemma~2, we know that for the root~$n$ there exists a local partial solution of
  the form~$\langle \emptyset, \emptyset, \emptyset \rangle$ for partial 
  solution~$\langle M, \sigma, \emptyset \rangle$ with $\sigma(r)=\infty$ for $r\in\prog$.  
  We already established the induction step in
  Proposition~\ref{prop:completeness}. 
  Hence, we obtain some (corresponding) tuples for every
  node~$t$. Finally, stopping at the leaves~$n$.
  In consequence, we have shown both soundness and completeness
  resulting in the fact that Theorem~\ref{thm:inc:correctness} is
  true.
\end{proof}

Theorem~\ref{thm:inc:correctness} states that we can decide the
problem~\AspCons by means of Algorithm~$\dpa_{\INC}$, which uses
Algorithm~\ref{fig:incinc}.

\subsection{Proof of Theorem~\ref{thm:inc:runtime} (Worst-case Runtime
  Bounds of \algo{INC})}\label{proof:inc:runtime}

First, we give a proposition on worst-case space requirements in
tables for the nodes of our algorithm.

\begin{proposition}\label{prop:inc:space}
  Given a program \prog, a TD ${\cal T} = (T, \chi)$ with
  $T=(N, \cdot, \cdot)$ of the incidence graph~$I(\prog)$, and a node
  $t\in N$. Then, there are at most
  $2^{k+1}\cdot\ell^{k+1}\cdot2^{2^{k+1}\cdot\ell^{k+1}}$ tuples in
  $\tab{t}$ using algorithm ${\dpa}_{\INC}$ for width $k$ of
  ${\cal T}$ and bound
  $\ell = max\{3, \bnd(r) : r \in \weight(\prog)\}$.
\end{proposition}
\shortversion{\begin{proof}{(Sketch)}}%
\longversion{\begin{proof}[Proof (Sketch).]}
  Let $\prog$ be the given program, ${\cal T} = (T, \chi)$ a TD of the
  incidence graph~$I(\prog)$, where $T=(N, \cdot, \cdot)$, and
  $t\in N$ a node of the TD. Then, by definition of a decomposition of
  the primal graph for each node~$t\in N$, we
  have~$\Card{\chi(t)} - 1 \leq k$.  In consequence, we can have at
  most $2^{k+1}$ many witnesses, and for each witness a subset of the
  set of witnesses consisting of at most~$2^{2^{k+1}}$ many
  counterwitnesses.  Moreover, we observe that
  Algorithm~\ref{fig:incinc} can be easily modified such that a state
  $\sigma: \prog_t \rightarrow \NAT_0 \cup \{\infty\}$ for node
  $t\in N$ assigns each weight rule $r\in \weight(\prog)$ a
  non-negative integer $\sigma(r) \leq \bnd(r) + 1$, each choice rule
  $r\in \choice(\prog)$ a non-negative integer $\sigma(r) \leq 2$ and
  each disjunctive rule $r \in \disj(\prog)$ a non-negative integer
  $\sigma(r) \leq 1$. This is the case since we need to model
  $\sigma(r) = 0$ and $\sigma(r) = \infty$ for each disjunctive rule
  $r$.  Moreover, for choice rules $r$, it suffices to additionally
  model whether $1 \leq \sigma(r) < \infty $, and for weight rules
  $r$, we require to remember any weight
  $1 \leq \sigma(r) \leq \bnd(r)$. In total, we need to distinguish
  $\ell^{k+1}$ different rule-states for each witness of a tuple in
  the table~$\tab{t}$ for node~$t$.  Since for each witness in the
  table~$\tab{t}$ for node~$t \in N$ we remember rule-states for at
  most~$k+1$ rules, we store up to $\ell^{k+1}$ many combinations per
  witness.  In total we end up with at most
  $2^{2^{k+1} \cdot \ell^{k+1}}$ many counterwitnesses for each
  witness and rule-state in the worst case. Thus, there are at most
  $2^{k+1} \cdot \ell^{k+1} \cdot 2^{2^{k+1}\cdot \ell^{k+1}}$ tuples
  in table~$\tab{t}$ for node~$t$. In consequence, we established the
  proposition.
\end{proof}

\shortversion{\begin{proof}[of Theorem~\ref{thm:inc:runtime}]}
\longversion{\begin{proof}[Proof of Theorem~\ref{thm:inc:runtime}.]}
  Let $\prog$ be a program, $I(\prog)=(V,\cdot)$ its incidence graph, and
  $k$ be the treewidth of $P(\prog)$.
  Then, we can compute in time~$2^{\bigO{k^3}} \cdot \Card{V}$ a TD of
  width at most~$k$~\citesec{Bodlaender96}.  We take such a TD and
  compute in linear time a nice TD~\citesec{Kloks94a}. Let
  $\TTT = (T,\chi)$ be such a nice TD with $T = (N,\cdot,\cdot)$.
  Since the number of nodes in~$N$ is linear in the graph size and
  since for every node~$t \in N$ the table~$\tab{t}$ is bounded by
  $2^{k+1} \cdot \ell^{k+1} \cdot 2^{2^{k+1}\cdot \ell^{k+1}}$
  according to Proposition~\ref{prop:inc:space}, we obtain a running
  time of~$\bigO{2^{2^{k+2}\cdot \ell^{k+1}}
    \CCard{I(\prog)}}$. Consequently, the theorem sustains.
\end{proof}

\subsection{Correctness of the Algorithm~$\dpa_{{\#O\INC}}$}

The following propositions states that we can use
Algorithm~$\dpa_{{\#O\INC}}$ to actually count optimal answer sets.

\begin{proposition}
  The algorithm~$\dpa_{{\#O\INC}}$ is correct.
\end{proposition}
\shortversion{\begin{proof}[Sketch]}
\longversion{\begin{proof}[Proof (Sketch).]}
  We follow the proof of Theorem~\ref{thm:inc:correctness}. First, we
  additionally need to take care of the optimization rules obtained by
  extending
  Definitions~\ref{def:partial-model}--\ref{def:localpartialsol}, the
  lemmas and propositions accordingly. In order to handle the
  counting, we have to extend
  Definitions~\ref{def:partial-model}--\ref{def:localpartialsol} by
  counters. Further, we additionally need to ensure and prove in the
  induction steps, which are established by
  Propositions~\ref{thm:soundness} and \ref{prop:completeness}, that
  any fixed partial solution is obtained from child to parent via a
  corresponding local partial solution by the
  algorithm. 
\end{proof}


\newpage
\section{Experiments}\label{appendix:experiments}

\subsection{Solvers}
The solvers tested include our own prototypical implementation, which
we refer to as DynASP, and the existing solvers
\begin{itemize}
\item Cachet~1.21~\cite{sat:SangBBKP04}, which is a SAT model counter,
\item \depqbfz\footnote{See
    \url{https://github.com/hmarkus/depqbf/tree/depqbf0}}, which is
  the solver DepQBF~\citesec{LonsingBiere10} where we added a naive
  implementation using methods described
  by~Lonsing~\citesec{Lonsing15},
\item Clasp~3.1.4~\citesec{GebserKaufmannSchaub12a}, which is an ASP
  solver, and
\item \sharpsat~12.08~\cite{sat:Thurley06}, which is a SAT model
  counter.
\end{itemize}

\subsection{Environment}
We ran the experiments on an Ubuntu 12.04 Linux cluster of 3 nodes
with two AMD Opteron 6176 SE CPUs of 12 physical cores each at 2.3Ghz
clock speed and 128GB RAM. Input instances were given to the solvers
via shared memory.  All solvers have been compiled with gcc version
4.9.3.  Available memory was limited to 4GB RAM, which was necessary
to run \sharpsat on larger instances, and CPU time to 300 seconds.
We used default options for cachet and \sharpsat, \mbox{``--qdc''} for
\depqbfz, ``--stats=2 --opt-mode=optN -n 0 --opt-strategy=usc -q'' and
no solution printing/recording for clasp. We also benchmarked clasp
with the flag ``bb''. However, ``usc'' outperformed ``bb'' on all our
benchmarks. All solvers have been executed in single core mode.

\subsection{Instances}
We used both random and structured instances for benchmark sets, which
we briefly describe below. The benchmark sets, including instances and
encodings, as well as results are available online on
github\footnote{See
  \url{https://github.com/daajoe/lpnmr17_experiments}.}.

The random instances (\pname{Sat-TGrid}, \pname{2QBF-TGrid},
\pname{ASP-TGrid}, \pname{2ASP-TGrid}) were designed to have a high
number of variables and solutions and treewidth at most three. 
The instances are constructed as follows:
%
Let $k$ and $\ell$ be some positive integers and $p$ a rational number
such that $0<p\leq 1$.
An instance~$F$ of \pname{Sat-TGrid$(k,l,p)$} consists of the
set~$V=\SB (1,1),\ldots, (1,\ell),(2,\ell),\ldots, (k,\ell) \SE$ of
variables and with probability~$p$ for each variable~$(i,j)$ such that
$1<i\leq k$ and $1<j\leq \ell$ a clause $s_1(i,j)$, $s_2(i - 1,j)$,
$s_3(i,j - 1)$, a clause $s_4(i,j)$, $s_5(i - 1,j)$, $s_6(i-1,j - 1)$,
and a clause $s_7(i,j)$, $s_8(i - 1,j - 1)$, $s_9(i,j - 1)$ where
$s_i \in \{-,+\}$ is selected with probability one half. In that way,
such an instance has an underlying dependency graph that consists of
various triangles forming for probability~$p=1$ a graph that has a
grid as subgraph.  Let $q$ be a rational number such that $0<q\leq
1$. An instance of the set~\pname{2Qbf\hy TGrid$(k,l,p,q)$} is of the
form~$\exists V_1. \forall V_2. F$ where a variable belongs to~$V_1$
with probability~$q$ and to $V_2$ otherwise. Instances of the
sets~\pname{ASP-TGrid} or \pname{2ASP-TGrid} have been constructed in
a similar way, however, as an \ASP program instead of a formula. Note
that the number of answer sets and the number of satisfiable
assignments correspond.  We fixed the parameters to $p=0.85$, $k=3$,
and $l\in \{40,80,\ldots,400\}$ to obtain instances that have with
high probability a small fixed width, a high number of variables and
solutions. Further, we took fixed random seeds and generated 10
instances to ensure a certain randomness.

The structured instances model various graph problems (\pname{2Col},
\pname{3Col}, \pname{Ds}, \pname{St} \pname{cVc}, \pname{sVc}) on real
world mass transit graphs of 82 cities, metropolitan areas, or
countries.  The graphs were extracted from publicly available mass
transit data feeds~\citesec{gtfs} using
gtfs2graphs~\citesec{Fichte16c} and split by transportation
type,~e.g., train, metro, tram. We excluded bus networks as size and
treewidth were too large.  For an input graph, the \pname{2Col}
encoding counts all minimal sets~$S$ of vertices s.t.\ there are two
sets~$F$ and $S$ where no two neighboring vertices~$v$ and $w$ belong
to~$F$; \pname{3Col} counts all 3-colorings; \pname{Ds} counts all
minimal dominating sets; \pname{St} counts all Steiner trees;
\pname{cVc} counts all minimal vertex covers; and \pname{sVc} counts
all subset-minimal vertex covers.
Since we cannot expect to solve instances of high treewidth efficiently, we
restricted the instances to those where we were able to find decompositions of
width below 20 within 60 seconds. 

\subsection{Extended Discussion on the Results}
In order to draw conclusions about the efficiency of our approach, we
mainly inspected the total cpu running time and number of timeouts on
the random and structured benchmark sets.  Note that we did not record
I/O times. The runtime for \dynasp{\cdot} includes decomposition times
using heuristics from~\citesec{Dell16b,DermakuEtAl08}. We randomly
generated three fixed seeds for the decomposition computation to allow
a certain variance~\cite{AbseherEtAl15a}. When evaluating the
results, we took the average over the three runs per instance.
Figure~\ref{fig:random} illustrates solver runtime on the various random
instance sets and a selected structured instance set as a cactus plot.
Table~\ref{tab:real_world} reports on the average running times, number of
solved instances, and number of timeouts of the solvers on the structured
instance sets.

\subsubsection{Results.}
%
%
%
\pname{SAT-TGrid} and \pname{Asp-TGrid}:
Cachet solved 125 instances.
Clasp always timed out. A reason could be the high number of solutions as
Clasp counts the models by enumerating them (without printing them).
\dynasp{\cdot} solved each instance within at most 270 seconds (on
average 67 seconds). The best configuration with respect to runtime
was \algo{PRIM}. However, the running times of the different
configurations were close. We observed as expected a sub-polynomial
growth in the runtime with an increasing number of solutions.
\sharpsat timed out on 3 instances and ran into a memory out on 7
instances, but solved most of the instances quite fast. Half of the
instances were solved within 1 second and more than 80\% of the
instances within 10 seconds, and about 9\% of the instances took more
than 100 seconds. The number of solutions does not have an impact on
the runtime of \sharpsat.
\sharpsat was the fastest solver in total. However, \dynasp{\cdot}
solved all instances. The results are illustrated in the two left
graphs of Figure~\ref{fig:random}.

\pname{2QBF-TGrid} and \pname{2ASP-TGrid}: Clasp solved more than half
of the instances in less than 1 second, however, timed out on 59
instances. \depqbfz shows a similar behavior as Clasp, which is not
surprising as both solvers count the number of solutions by
enumerating them and hence the number of solutions has a significant
impact on the runtime of the solver. However, Clasp is faster
throughout than \depqbfz.
DynASP2(\algo{INC}) solved half of the instances within less than 1
second, about 92\% of the instances within less than 10 seconds, and
provided solutions also if the instance had a large number of answer
sets. DynASP2(\algo{PRIM}) quickly produced timeouts due to large
rules in program that produced a significantly larger width of the
computed decompositions.

Structured instances: 
Clasp solved most of the structured instances reasonably
fast. However, the number of solutions has again, similar to the
random setting, a significant impact on its performance. If the number
of solutions was very high, then Clasp timed out. If the instance has
a small number of solutions, then Clasp yields the number almost
instantly. However, \dynasp{\cdot} also provided a solution within a
second.
\dynasp{\cdot} solved for each set but the set~\pname{St} more than
80\% of the instances in less than 1 second and the remaining
instances in less than 100 seconds. For \pname{St} the situation was
different. Half of the instances were solved in less than 10 seconds
and a little less than the other half timed out.  Similar to the
random setting, \dynasp{\cdot} ran still fast on instances with a
large number of solutions.

\bibliographystylesec{plain}
\bibliographysec{references_appendix}

\end{document}